\documentclass[12pt, reqno]{amsart}
\usepackage[utf8]{inputenc}
\usepackage[english]{babel}
\usepackage{amssymb}
\usepackage[foot]{amsaddr}

\usepackage{comment}
\usepackage[normalem]{ulem}

\usepackage{blkarray}
\usepackage{fullpage}              
\usepackage{tikz}             
\usetikzlibrary{arrows.meta,positioning,calc,decorations.pathreplacing}
\usetikzlibrary{plotmarks}
\usetikzlibrary{shapes}
\usetikzlibrary{arrows}
\usepackage{graphicx}
\usepackage{caption}
\usepackage{subcaption}
\usepackage[pdftex]{hyperref}
\hypersetup{
    colorlinks=true,
    linkcolor={MidnightBlue},
    citecolor={magenta},
    urlcolor={MidnightBlue}
}
\usepackage[hang,flushmargin]{footmisc}
\usepackage[ruled,vlined]{algorithm2e}
\usepackage[nameinlink,capitalise]{cleveref}
\usepackage{soul}

\SetKwComment{Comment}{/* }{ */}

\numberwithin{equation}{section}
\theoremstyle{plain}
\newtheorem{theorem}{Theorem}[section]

\newtheorem{lemma}[theorem]{Lemma}
\newtheorem{proposition}[theorem]{Proposition}

\theoremstyle{definition}
\newtheorem{remark}[theorem]{Remark}
\newtheorem{examplex}[theorem]{Example}
\newenvironment{example}
  {\pushQED{\qed}\examplex}
  {\popQED\endexamplex}

\newtheorem{definition}[theorem]{Definition}

\newcommand{\Cc}{{\mathbb C}}

\newcommand{\Rr}{{\mathbb R}}
\newcommand{\Zz}{{\mathbb Z}}
\newcommand{\Pp}{{\mathbb P}}

\newcommand{\cf}{\mathcal{F}}

\newcommand{\cn}{\mathcal{N}}
\newcommand{\ct}{\mathcal{T}}

\newcommand{\ci}{\mathcal{I}}

\newcommand{\cl}{\mathcal{L}}

\DeclareMathOperator{\desc}{desc}

\definecolor{joe}{RGB}{180, 0, 240}

\definecolor{bryson}{RGB}{240,0,0}

\definecolor{cyan}{RGB}{0,240,250}
\definecolor{MidnightBlue}{RGB}{25,25,112}

\title{Identifiability of phylogenetic networks and quintet concordance factors}

\author{Joseph Cummings$^{*,1}$}
\thanks{*Corresponding author}
\address{1 University of Edinburgh}
\author{Maize Curiel}
\author{Bryan Currie$^2$}
\address{2 New Jersey Institute of Technology}
\author{Bryson Kagy$^3$}
\address{3 Texas State University}
\author{Udani Ranasinghe$^4$}
\address{4 University of Hawai`i at M\=anoa}
\author{John A. Rhodes$^5$}
\address{5 University of Alaska Fairbanks}
\date{\today}


\begin{document}

\begin{abstract}
    \noindent 
    Several statistical methods of phylogenetic network inference and testing for non-tree-like relationships are based on assessing genomic data through quartet Concordance Factors, the frequencies of 4-taxon topological relationships on gene trees. While such an approach obviates making several undesirable modeling assumptions, it also results in non-identifiability issues for network roots and for small cycles.
    In this work, an algorithm and accompanying Macaulay2 implementation are provided for computing $n$-tet Concordance Factors on any phylogenetic network. 
    We employ this algorithm on quintet Concordance Factors, summarizing 5-taxon gene trees, to explore identifiability of level-1 networks under the Network Multispecies Coalescent model.
    We show some additional network features become identifiable that are not through quartets.  
    As identifiability is a necessary prerequisite to inference by any method, this lays a foundation for future inference work.  
\end{abstract}

\maketitle

\section{Introduction}
\label{sec: intro}

Although phylogenetic trees are the primary means of depicting relationships between species or populations, analyses of genomic data has led to to increasing recognition of their limitations. Growing evidence for hybrid speciation, admixture, or other types of lateral gene flow suggests that \emph{phylogenetic networks} are often needed.

Inference of phylogenetic networks, however, remains challenging. For instance, we now understand that viewing this problem as a 2-stage one, in which a `species tree' is first inferred, and later analyses add additional edges to represent gene flow, can be misleading \cite{Solis-LemusEtAl2016,DinhBanos2025}. The presence of gene flow means a tree model is misspecified, and when common methods for inferring one are applied the result may be a tree not even displayed on the network. Adding additional edges cannot correct such a mistake.

However, attempts
to directly infer a network of arbitrary structure face other problems. The vastness of network space dwarfs that of tree space, so search-based methods can be computationally infeasible without imposing some limit on network structure. To account for incomplete lineage sorting (ILS) the Network Multispecies Coalescent (NMSC) model \cite{Meng2009} is used, introducing a difficult translation between metric information on gene trees (in substitution units) and those on the species network (in coalescent units). Computation of the full likelihood function for more than a handful of taxa adds an additional major burden. Tractable current inference methods focus on a limited class of networks (e.g., level-1, or with at most $k$ reticulations), using gene tree summaries with pseudo-likelihood or combinatorial ideas as part of novel inference approaches. \cite{YuNakhleh2015,Chifman2015b,Solis-LemusEtAl2016, NANUQ_2019, NANUQ+, HoltegrefeEtAl2025, PhyNEST2025}.

\medskip

Underneath all of this, though, is the more fundamental issue of \emph{identifiability}. A model feature or parameter of interest (e.g., the full network, the hybrid node in a cycle on a network, the network's root location) is  said to be \emph{identifiable} if the feature's presence or the parameter's value is determined by a distribution of data the model entails. While identifiability says nothing about a particular inference method, it is a prerequisite for
any method to be statistically consistent.

The counts of displayed 4-taxon unrooted topological trees across a collection of (inferred) gene trees form a data type, the empirical \emph{quartet Concordance Factors} (CFs), that has been found particularly useful in inference of species trees and networks under the coalescent model. Their use avoids rooting gene trees (which may be error prone in the presence of ILS) or relating metric units on gene trees to those on species trees or networks (using poorly-justified \emph{ad hoc} assumptions). It has also been argued that CF-based methods offer  robustness to violations of an assumption of no intragenic recombination \cite{Rivas-GonzalesEtAl2026}.

Network identifiability results for quartet CFs under the NMSC model include unrooted species trees \cite{Allman2011} and certain features of both level-1 and more general networks \cite{SNaQ, Banos2019, AllmanEtAl2024, ECTOBlob}. These theoretical results undergird the software ASTRAL \cite{ASTRAL}, SNaQ \cite{SNaQ}, NANUQ \cite{NANUQ_2019}, NANUQ$^+$ \cite{NANUQ+}, TINNiK \cite{TINNIK}, ECTOBlob \cite{ECTOBlob},  TREE-QMC \cite{TREEQMC} and TOB-QMC \cite{TOBQMC}.

However, some topological features even of level-1 networks are \emph{not} identifiable from quartet CFs. These include the root location, and certain features of small cycles: the presence of 2-cycles, and which node is hybrid in a 3- or 4-cycle.
Investigating quintet CFs, the frequencies of 5-taxon unrooted topological relationships displayed on gene trees, is a natural next step to further understand feature identifiability.
That quintet CFs allow for species tree root identification was shown in \cite{Allman2011}, and developed into an inference method in \cite{Tabatabaee22, Tabatabaee23}, so one might expect an analogous result for networks. The 3-cycle identifiability results of
\cite{AllmanEtAl2024} for quartet CFs were found by studying non-linear relationships between quartet CFs across 5-taxon subtrees, again suggesting quintets offer more information than quartets.

\medskip

This work begins a study of quintet CFs for networks,  focusing on the question of root and 3-cycle identifiability. 
Symbolic computation is a key aspect of this investigation. Although formulas for quintet CFs, in terms of branch lengths and hybridization parameters, were given for species trees in \cite{Allman2011}, formula complexity grows considerably for networks. 
As a first contribution, we therefore present and implement a new algorithm to produce such symbolic formulas. With flexibility for future investigations in mind, it accepts an arbitrary rooted network and an unrooted gene tree on any subset of the network's taxa, returning the probability of that  gene (sub)tree under the NMSC, as a polynomial function of edge probabilities (i.e., transformed edge lengths) and hybridization parameters.

To address identifiability questions for 5-taxon networks, we view these symbolic formulas as parameterizing a portion of a \emph{quintet CF algebraic variety}. If the varieties for two networks differ, then we can distinguish between expected quintet CFs from the two for generic parameter values. We therefore undertake an exhaustive study of all quintet CF varieties for level-1 networks with no 2-cycles, classifying networks according to them. We omit consideration of 2-cycles here (since they lead to infinitely many networks), though a subsequent work will extend these results to fully address them. The amount of computation is reduced by first classifying networks by information known to be identifiable from quartet CFs. Nonetheless, computing some of the varieties is beyond the limits
of current software, so several alternative strategies are used to obtain our results.

\medskip
In \Cref{sec: background} we recall basic terminology and definitions. 
A recursive algorithm for computing symbolic formulas for quintet CFs (or more generally $m$-tet CFs) in terms of network edge probabilities and hybridization parameters is presented in \Cref{sec: algorithm}, with a Macaulay2 \cite{M2} implementation made available. 
\Cref{sec: root identifiability} contains both theorems on non-identifiability of certain  structures on 5-taxon networks, and exhaustive results from computations of quintet CF invariants for all 5-taxon rooted level-1 networks without 2-cycles. We end with final comments in \Cref{sec: discussion}.

\section{The Network Multi-species Coalescent Model and Concordance Factors}
\label{sec: background}
We first describe the setting for our work and  fix terminology.

\subsection{Phylogenetic networks}
A \emph{rooted phylogenetic network} $\cn=\cn^+$ on a taxon set $X$ is a rooted, connected, directed acyclic graph (DAG) whose leaves are bijectively labeled by elements of $X$. We assume the root of the network is also the least stable ancestor (LSA) of the taxa, since network structure above the LSA has no impact on topological gene trees under the coalescent model.
We consider primarily \emph{binary} networks, meaning the root has out-degree 2, leaves have in-degree 1, and all other nodes have degree 3. \emph{Hybrid nodes} are those with in-degree $\ge 2$. \emph{Hybrid edges} are those whose child node is hybrid. All other nodes and edges are \emph{tree nodes} and  \emph{tree edges}. The set of taxa below a node $v$ is denoted $\desc_X(v)$. A binary network is \emph{level-1} if each of its 2-edge-connected components has at most one hybrid node, or equivalently, when the network edges are undirected all of its  cycles are disjoint.

A \emph{semidirected} phylogenetic network is obtained from a directed one by undirecting all tree edges, retaining directions of hybrid edges, and suppressing  the root (if of degree 2).
A \emph{blob} in a network is a maximal 2-edge connected component. The \emph{tree of blobs} of a network is obtained by contracting each blob in its semidirected form to a node, and suppressing any degree 2 nodes, leaving an unrooted tree. For a binary network, multifurcations (nodes of degree $\ge4$) in the tree of blobs represent unspecified blob structures, while nodes of degree 3 may represent either trivial (single node) blobs or more complex one. Blobs with  single in- and out-nodes are lost in the tree of blobs.

\subsection{The coalescent model}\label{subsec:coalescent intro}
To parameterize the \emph{network multispecies coalescent (NMSC) model} with independent inheritance on a fixed binary $\cn$, the network is endowed with \emph{edge lengths} $x_e\in (0,\infty)$ for all edges $e$. 
Additionally, hybrid edges $h_1,h_2$ are given \emph{hybridization parameters} $\lambda_{h_i}\in(0,1)$ with $\lambda_{h_1}+\lambda_{h_2}=1$ if hybrid edges $h_1,h_2$ share a common child hybrid node. This associates to the topological network a stochastic parameter space $\Theta_\cn$, which is an open subset of $\mathbb R^{n_t+3n_{h}/2}$ with $n_t,n_h$ the number of tree and hybrid edges. There are 3 parameters (2 edge lengths, one hybridization parameter) for each pair of hybrid edges. Edge lengths here are in \emph{coalescent units}, a measure of number of generations divided by population size, so that the coalescent rate of 2 lineages traced backward within a common edge (population) of the species network is 1. The coalescence of each such pair of lineages is independent. Lengths of pendant edges may be omitted if only one sample per taxon is taken, since then no coalescence can occur in that edge. Hybridization parameters represent the probabilities that if a lineage, traced backwards in time, reaches a hybrid node it then enters a specific hybrid edge, independently of other lineages. 

Since not all lineages will have necessarily coalesced by the time they reach the network root, we introduce an infinite length edge with its child as the network root so that all lineages will coalesce with probability 1.

The NMSC model assigns to a fixed network $\cn$ (for example, \Cref{fig1}) with stochastic parameters a probability distribution for metric rooted gene trees, as a collection of lineages, one from each taxon, undergoes the coalescent process. For  more detail on the NMSC model see \cite{Meng2009}.

\begin{figure}[ht]
    \centering
    \begin{tikzpicture}[
        >=Stealth,
        every node/.style={font=\small},
        dot/.style={circle,fill=black,inner sep=1.3pt},
        leaf/.style={dot, label=below:{\small #1}},
        int/.style={dot},
    ]

        \node[int] (u0) at (0,4) {};
        \node[int] (u1) at (3,2) {};
        \node[int] (u2) at (-2,3) {};
        \node[int] (u3) at (-1.5,2.5) {};
        \node[int] (u4) at (-1,2) {};
        \node[int] (u5) at (-2,1) {};
        \node[int] (u6) at (-3,2) {};

        \node[leaf=$A$] (A) at (-4,0) {};
        \node[leaf=$B$] (B) at (-2,0) {};
        \node[leaf=$C$] (C) at (0,0) {};
        \node[leaf=$D$] (D) at (2,0) {};
        \node[leaf=$E$] (E) at (4,0) {};

        \draw[->,line width=0.8pt] (u0) -- node[midway, above] {$\ell_1$} (u2);
        \draw[->,line width=0.8pt] (u0) -- node[midway, above] {$\ell_8$} (u1);
        \draw[->,line width=0.8pt] (u2) -- node[midway, right] {$\ell_2$} (u3);
        \draw[->,line width=0.8pt] (u2) -- node[midway, above left] {$\ell_3$} (u6);
        \draw[->,line width=0.8pt] (u6) -- node[midway, below left] {$\ell_5$} (u5);
        \draw[->,line width=0.8pt] (u3) -- node[midway, right] {$\ell_6$} (u4);
        \draw[->,line width=0.8pt] (u4) -- node[midway, below right] {$\ell_7$} (u5);
        \draw[->,line width=0.8pt] (u3) -- node[midway, below right] {$\ell_4$} (u6);
        \draw[->,line width=0.8pt] (u6) -- (A);
        \draw[->,line width=0.8pt] (u5) -- (B);
        \draw[->,line width=0.8pt] (u4) -- (C);
        \draw[->,line width=0.8pt] (u1) -- (D);
        \draw[->,line width=0.8pt] (u1) -- (E);

    \end{tikzpicture}
\caption{A level-2 phylogenetic network with edge length parameters $\ell_i$. If $\lambda_{1}$ is the probability that a lineage enters the edge of length $\ell_3$, then the probability that same lineage enters the edge of length $\ell_4$ is $1-\lambda_1$. Similarly, if $\lambda_2$ is the probability a lineage enters the edge of length $\ell_5$, then the probability of entering the edge of length $\ell_7$ is $1-\lambda_2$.}
\label{fig1} 
\end{figure}
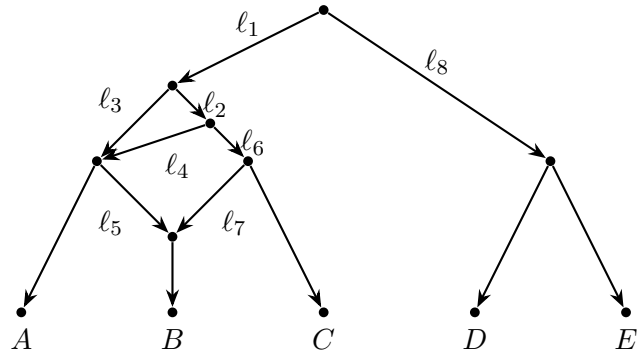

\subsection{Concordance factors}

Marginalizing a gene tree distribution over branch lengths and root locations gives a probability distribution on unrooted topological gene trees. Under the NMSC the fully resolved topological gene trees are exactly those with positive probability. 
Thus for an $n$-taxon network, by fixing some ordering of the $(2n-5)!!$ binary unrooted topological gene trees on $n$ taxa, we have a map 
$$\phi_\cn:\Theta_\cn\to \Delta_{(2n-5)!!-1}$$
from the stochastic parameter to the probability simplex of dimension $(2n-5)!!-1$. 
For a metric network $(\cn,\theta)$ with $\theta \in \Theta_\cn$, the image $\phi_\cn(\theta)$ is the \emph{gene tree concordance factor (CF) vector}. 
If $Y$ is a subset of taxa of size $m \geq 4$, then marginalizing over the remaining taxa in $X \smallsetminus Y$, one obtains 
$$\phi^Y_\cn:\Theta_\cn \to \Delta_{(2m-5)!!-1},$$
whose values are the $m$-tet CF for the set $Y$ on $\cn$. 
(Due to the structure of the NMSC model, these can also be viewed as concordance factors for the induced metric network $\cn|_Y$.) 
The collection of $m$-tet concordance factors for all such taxon subsets $Y\in\binom {X}{m}$ forms the $m$-tet CFs for $(\cn,\theta$). 
Taking $m=4$ gives the quartet CFs which have been well studied, and $m=5$ give the quintet CFs we focus on.

\ 

How the component functions of $\phi$ can be computed will be discussed in the next section, but an important observation is that edge lengths $x$ only appear through exponential expressions of the form
$$\exp(-kx)$$
with $k$ an integer. As a result, transforming parameters as $X=\exp(-x)$
results in $\phi$ being a polynomial map. This transformed length 
is called an \emph{edge probability}  since it gives the probability that two lineages entering a common edge of length $x$ and tracing backwards in time reach its end without coalescing. For convenience, we will henceforth work with edge probability parameters  rather than edge lengths, replacing the original parameter space with the transformed one.

\subsection{Algebra, Distinguishability, and Identifiability}

Considering $\phi^Y_\cn$ as a polynomial map, we can use algebra to answer identifiability questions from $m$-tet CFs. For this, we first assume $X=Y$, $|X|=m$
and for each rooted network
$\cn$ on $X$ consider
the polynomial map 
\[
    \phi_\cn=\phi_\cn^X:\Theta_\cn \to \Delta_{(2m-5)!! -1},
\]
which extends from the stochastic numerical parameters space $\Theta_\cn$ to a  polynomial map on a complex space
\[
    \Phi_\cn: \Cc^{N} \to \Cc^{(2m-5)!!}
\]
with $N$ the number of numerical parameters for $\cn$.
If $\{u_T\}$ is a collection of indeterminates indexed by all binary unrooted topological gene trees $T$ on $X$, this map induces a ring homomorphism
\begin{align*}
    \Phi_\cn^* : \Cc[\{u_T\}] &\to \Cc[\theta_1,\dotsc,\theta_N] \\
    u_T &\mapsto \Phi_\cn(\theta)_{T}
\end{align*}
whose kernel (i.e., all polynomials in the $u_T$ which vanish on all CFs, even for the original parameter space $\Theta_\cn$) is an ideal $I_\cn$.
Elements of $I_\cn$ are called \emph{phylogenetic CF invariants} of $\cn$, and express the polynomial relationships among the CF entries.

The locus of points in $\mathbb C^N$ at which all
elements of $I_\cn$ vanish
is the \emph{CF variety} of $\cn$, denoted $V_\cn = V(I_\cn)$. The set of CF probability distributions $\phi(\Theta_{\cn})$ is a subset of $V_\cn$, which is its
Zariski closure.

CF ideals and varieties provide a powerful tool for  distinguishing between
networks, by capturing structure across all parameter values. 
Suppose $\cn_1$ and $\cn_2$ are two network topologies. 
If $V_{\cn_1} \neq V_{\cn_2}$, then from basic algebraic geometry either
the varieties intersect in one of lower dimension, or one is a proper subvariety of the other. If their intersection is of lower dimension, then for all numerical parameters (even from $\Theta_{\cn_i}$) outside a set of measure zero CFs on one will not lie in the other. 
If $V_{\cn_1}\subsetneq V_{\cn_2}$, then, again outside a set of measure zero, parameters on $\cn_2$ give CFs not on $V_{\cn_1}$. Thus we say in either case that for \emph{generic parameters} the networks are
\emph{distinguishable by CFs}. 
If all pairs of networks within some class are generically distinguishable by CFs, then we say that within that class the networks are \emph{generically identifiable by CFs}.

In \Cref{sec: root identifiability}, we sort the class of 5-taxon 2-cycle free level-1 networks based on computations of their quintet CF ideals $I_\cn$, and hence their quintet varieties, $V_\cn$.
The first step toward this is undertaken in the next section, where we develop an algorithm for computing the map $\phi_\cn$.

We end this section with an example illustrating the power
of quintet CF varieties for root identifiability.

\begin{example}
\label{example: quartet-quintet-example}

    We recast some of the results of \cite{Allman2011}, on identifiability of roots from CFs for phylogenetic trees, in the algebraic framework just described. Specifically, we
    show that quartet CFs fail to provide enough information to capture the tree's root location, while quintet CFs allow one to identify the root.
    The parameterizations used below were computed using the algorithm described in \Cref{sec: algorithm}, though they can also be found in \cite{Allman2011}.

    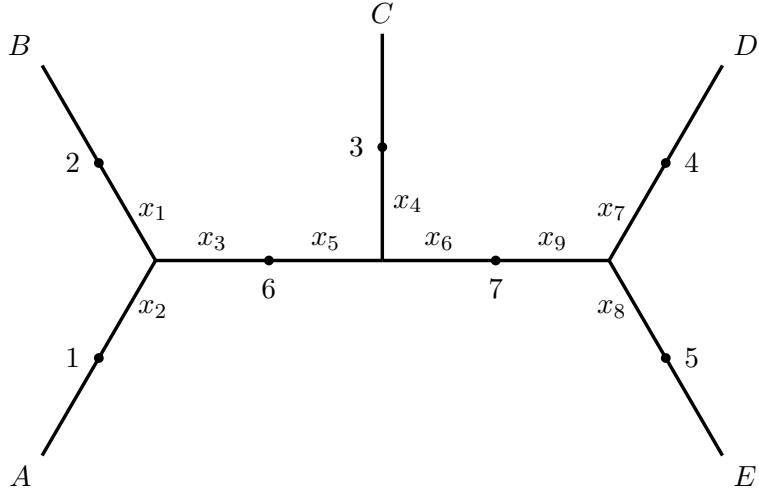
\begin{figure}
    \begin{center}
    \begin{tikzpicture}[
        >=Stealth,
        every node/.style={font=\small},
        dot/.style={circle,fill=black,inner sep=1.3pt},
        leaf/.style={dot, label=below:{\small #1}},
        int/.style={dot},
        scale=0.3
    ]
    \coordinate (u) at (0,0) {};
    \coordinate (v) at (-10,0) {};
    \coordinate (w) at (10,0) {};
    \coordinate[label=below left:$A$] (A) at (-15,-8.6) {};
    \coordinate[label=above left:$B$] (B) at (-15,8.6) {};
    \coordinate[label=above:$C$] (C) at (0,10) {};
    \coordinate[label=above right:$D$] (D) at (15,8.6) {};
    \coordinate[label=below right:$E$] (E) at (15,-8.6) {};

    \coordinate[label=right:$x_1$] (x1) at (-11.25,2.15) {};
    \coordinate[label=right:$x_2$] (x2) at (-11.25,-2.15) {};
    \coordinate[label=above:$x_3$] (x3) at (-7.5,0) {};
    \coordinate[label=above:$x_5$] (x5) at (-2.5,0) {};
    \coordinate[label=right:$x_4$] (x4) at (0,2.5) {};
    \coordinate[label=above:$x_6$] (x6) at (2.5,0) {};
    \coordinate[label=above:$x_9$] (x9) at (7.5,0) {};
    \coordinate[label=left:$x_7$] (x7) at (11.25,2.15) {};
    \coordinate[label=left:$x_8$] (x8) at (11.25,-2.15) {};
    
    \draw[line width=1.3pt] (A) -- node[int,midway,label=left:$1$] {} (v);
    \draw[line width=1.3pt] (B) -- node[int,midway,label=left:$2$] {} (v);
    \draw[line width=1.3pt] (C) -- node[int,midway,label=left:$3$] {} (u);
    \draw[line width=1.3pt] (D) -- node[int,midway,label=right:$4$] {} (w);
    \draw[line width=1.3pt] (E) -- node[int,midway,label=right:$5$] {} (w);
    \draw[line width=1.3pt] (v) -- node[int,midway,label=below:$6$] {} (u);
    \draw[line width=1.3pt] (u) -- node[int,midway,label=below:$7$] {} (w);
    \end{tikzpicture}
    \end{center}
    \caption{The 5-taxon unrooted topological tree, with  its seven possible root locations and edge probabilities labelled.}\label{fig:5tree}
    \end{figure}

    Consider the unrooted phylogenetic tree $\ct$ of \Cref{fig:5tree}, on taxa  $X = \{A,B,C,D,E\}$. It can be rooted at any of the seven numbered nodes to create a rooted binary phylogenetic tree. 
    Let $\ct^+_i$ denote the tree whose root is at position $i$.
    Note the edge probabilities for each edge are given in \Cref{fig:5tree}, 
    and degree 2 nodes are not suppressed
    so that each edge probability makes sense 
    regardless of the root location.
    
    Fixing $i$, for each 4-taxon subset $Y \in \binom{X}{4}$, we have a map for the quartet CFs
    \[
        \phi_{\ct^+_i}^Y : \Theta_{\ct_i^+} \to \Delta_2,
    \]
    and we can consider all quartet CFs at once with
    \begin{align*}
        \phi_{\ct^+_i}^{\text{quartet}} : \Theta_{\ct_i^+} &\to (\Delta_2)^5 \\
        \theta &\mapsto (\phi_{\ct^+_i}^Y(\theta))_{Y \in \binom{X}{4}}
    \end{align*}
    We label coordinates of $(\Delta_2)^5 \subset [0,1]^{15}$ by the splits of the corresponding unrooted 4-taxon gene tree, and compute the vanishing ideal of the Zariski closure of the image of $\phi_{\ct^+_i}^{\text{quartet}}$.
    Regardless of $i$, the parameterizations are actually identical;
    they are defined as follows.
    \[
    \begin{array}{lll}
        u_{\texttt{AB|CD}} \mapsto 1 - \frac{2}{3} x_3 x_5, &
        u_{\texttt{AC|BD}} \mapsto \frac{1}{3} x_3 x_5, & 
        u_{\texttt{AD|BC}} \mapsto \frac{1}{3}x_3 x_5 \\[4pt]
        u_{\texttt{AB|CE}} \mapsto 1 - \frac{2}{3} x_3x_5, &
        u_{\texttt{AC|BE}} \mapsto \frac{1}{3} x_3 x_5, &
        u_{\texttt{AE|BC}} \mapsto \frac{1}{3} x_3 x_5 \\[4pt]
        u_{\texttt{AB|DE}} \mapsto 1 - \frac{2}{3} x_3 x_5 x_6 x_9, &
        u_{\texttt{AE|BD}} \mapsto \frac{1}{3} x_3 x_5 x_6 x_9, &
        u_{\texttt{AD|BE}} \mapsto \frac{1}{3} x_3 x_5 x_6 x_9 \\[4pt]
        u_{\texttt{AE|CD}} \mapsto \frac{1}{3} x_6 x_9, &
        u_{\texttt{AC|DE}} \mapsto 1 - \frac{2}{3}x_6 x_9, &
        u_{\texttt{AD|CE}} \mapsto \frac{1}{3} x_6 x_9 \\[4pt]
        u_{\texttt{BE|CD}} \mapsto \frac{1}{3} x_6 x_9, & 
        u_{\texttt{BD|CE}} \mapsto \frac{1}{3}x_6 x_9, & 
        u_{\texttt{BC|DE}} \mapsto 1 - \frac{2}{3} x_6 x_9
    \end{array}
    \]
    The kernel of the map can be readily computed, and a generating set is given below.
    \[
    \begin{array}{ccc}
    u_{\texttt{BE|CD}}+u_{\texttt{BD|CE}}+u_{\texttt{BC|DE}}-1, & u_{\texttt{AE|CD}}+u_{\texttt{AC|DE}}+u_{\texttt{AD|CE}}-1, & u_{\texttt{AB|DE}}+u_{\texttt{AE|BD}}+u_{\texttt{AD|BE}}-1, \\
    u_{\texttt{AB|CE}}+u_{\texttt{AC|BE}}+u_{\texttt{AE|BC}}-1, & u_{\texttt{AB|CD}}+u_{\texttt{AC|BD}}+u_{\texttt{AD|BC}}-1, & u_{\texttt{AB|CD}}-u_{\texttt{AB|CE}},\\
    u_{\texttt{AC|BD}}-u_{\texttt{AD|BC}}, & u_{\texttt{AC|BD}}-u_{\texttt{AC|BE}}, & u_{\texttt{AC|BD}}-u_{\texttt{AE|BC}},\\
    u_{\texttt{AE|BD}}-u_{\texttt{AD|BE}}, & u_{\texttt{AC|DE}}-u_{\texttt{BC|DE}}, & u_{\texttt{AD|CE}}-u_{\texttt{BD|CE}},\\ 
    u_{\texttt{BE|CD}}-u_{\texttt{BD|CE}}, & u_{\texttt{AE|CD}}-u_{\texttt{BD|CE}}, & 3u_{\texttt{AE|BC}}u_{\texttt{AE|CD}}-u_{\texttt{AE|BD}}.
    \end{array}
    \]
    While this generating set is not minimal, it is simple to describe: 
    The linear invariants are given by the five sum-to-1 conditions for each 4-taxon subset, and by nine binomials induced by the cherry symmetries $A \leftrightarrow B$ and $D \leftrightarrow E$ of $\ct$.
    The quadratic generator, when expressed in terms of the parameters, states that the edge probabilities for the two internal edges of $\ct$ multiply to give the edge probability for the composite edge they form.
    
    This ideal defines a degree 2 surface (of dimension 2) in $\Rr^{15}$.
    If one were only given the ideal, one might suspect that quartet CFs
    are not enough to generically identify the root. 
    However, it is possible that there could be additional semialgebraic
    constraints which could be used to identify the root.
    In our case, since the parameterizations are the same regardless of root location, 
    this cannot happen, and the semialgebraic sets $\phi_{\ct_i}^{\text{quartet}}\left((\Delta_2)^5\right)$ 
    (for $i = 1,\dotsc,7$)
    are identical.
    In particular, the root is not identifiable from quartet CFs.

    In contrast, instead of considering quartets, we can consider quintet CFs from the maps
    \[
        \phi_{\ct_i^+} : \Theta_{\ct_i^+} \to \Delta_{14}.
    \]
    Coincidentally, the image of this map lives in $\Rr^{15}$, though the coordinates are now labeled by unrooted 5-taxon gene trees.
    For each $i$, we can compute the vanishing ideal of the Zariski closure of the image of $\phi_{\ct_i^+}$.
    This computation results in seven distinct ideals, each defining a variety of dimension 3.
    Therefore, we can conclude that quintet CFs generically identify the location of the root of phylogenetic trees.
\end{example}

\subsection{Algebraic Matroids}
In \Cref{example: quartet-quintet-example},
it was possible to completely characterize several CF ideals.
However, for more complex parameterizations it may be infeasible to compute a full (or even partial) generating set.
In such cases, algebraic matroids can be a powerful tool for distinguishing two parameterized varieties \cite{HolleringSullivant2021}. 
As such, we provide a definition of a matroid here and the necessary theory for using matroids for distinguishability.

\begin{definition}
\label{def: matroid}
    A \textit{matroid} $\mathcal{M} = (E, \ci)$ is a pair with a finite set, $E$, and a collection $\ci \subseteq 2^E$ of subsets of $E$, satisfying the following conditions:
    \begin{enumerate}
        \item $\emptyset \in \ci$,
        \item if $I' \subseteq I \in \ci$, then $I' \in \ci$, and
        \item if $I_1,I_2 \in \ci$ with $|I_1| < |I_2|$, then there exists $e \in I_2 \setminus I_1$ so that $I_1 \cup \{e\} \in \ci$.
    \end{enumerate}
    A set $I \in \ci$ is said to be \textit{independent}.
\end{definition}

A proof of the proposition below can be found in \cite{rosen2014computing}.
\begin{proposition}
\label{prop: alg-matroid}
Let $k$ be a field.
Let $V \subseteq k^n$ be an irreducible variety with prime vanishing ideal $P \subseteq k[x_1,\dotsc,x_n]$.
Then $\mathcal{M}(V)$ is a matroid on $\{1,\dotsc,n\}$ defined by $I$ is independent if and only if 
\[
    P \cap k[x_i ~|~ i \in I] = \langle 0 \rangle.
\]
Moreover, suppose $\phi : k^m \to k^n$ is a polynomial map of the form 
\[
    \phi(\theta_1,\dotsc,\theta_m) = (\phi_1(\theta),\dotsc,\phi_n(\theta))
\]
with $V = \overline{\phi(k^m)}^{\text{Zar}}$, and
consider the transpose of the Jacobian matrix
\[
    J(\phi) = \left(\frac{\partial \phi_j}{\partial \theta_i}  \right), \; 1\leq i \leq m,\; 1 \leq j \leq n.
\]
Then $\mathcal{M}(V)$ can be described as follows: a subset $I \subseteq \{1,\dotsc,n\}$ is independent if and only if the columns of $J(\phi)$ indexed by the elements of $I$ are linearly independent over the fraction field $k(\theta_1,\dotsc,\theta_m)$.
\end{proposition}

The power of \Cref{prop: alg-matroid} in this work is that it allows us to distinguish between two CF varieties by using their parameterizations to show they have distinct matroids without any knowledge of the invariants.
Indeed, if two networks have distinct algebraic matroids,
then the ideals must be different.
For us, it will be enough to find a single circuit (i.e., a minimal dependent set) for one matroid that is not a circuit for the other matroid.

\section{Concordance Factor Algorithm}
\label{sec: algorithm}

In this section, we provide the mathematical framework
for our algorithm which computes symbolic formulas for the distribution of
gene trees for a species network, i.e., the concordance factors. Although ultimately interested in quintet CFs for binary networks, the algorithm does not require these special assumptions.

\subsection{Background} The general computation of probabilities of rooted topological gene trees from species trees or networks under a coalescent model began with
\cite{Degnan2005}. There  
a \emph{coalescent history}
for a fixed species tree $\ct$ and rooted topological gene tree $T^+$ is defined as a pairing of the nodes of $T^+$ with certain edges (populations) of $\cn$ in which
coalescent events resulting in those nodes could occur as $T^+$ is realized under the NMSC. 
Not all pairings are coalescent histories for $T^+$, as coalescences can occur only in edges ancestral to the taxon lineages involved, and the structure of $T^+$ may impose additional constraints. 

The probability of a coalescent history $\mathbf{h}$ leading to $T^+$ under the multispecies coalescence on $\ct$ is then 
\begin{equation}
    c(\mathbf{h}) \prod_{e \text{ edge of }\cn} g_{i(\mathbf{h},e), j(\mathbf{h},e)} (x_e),
\end{equation}
where $c(\mathbf{h})$ is a rational constant reflecting certain multiplicities, 
$j(\mathbf{h},e)$ is the number of lineages leaving the edge $e$ after the coalescent process starting with
$i(\mathbf{h},e)$ lineages entering the edge $e$. The functions $g_{i,j}(x)$ give the probability that $i$ lineages coalesce to $j$ lineages in an edge of length $x$, with formulas found by \cite{Tavare1984}:
\begin{equation}
    g_{i,j}(t) = \sum_{k=j}^i \mathrm{exp}\left(-\binom{k}{2} t\right) \frac{(2k-1)(-1)^{k-j}}{j!(k-j)!(j+k-1)} \prod_{m=0}^{k-1}\frac{(j+m)(i-m)}{i+m}, \quad 1 \le j \le i
\label{eq:Tavare}
\end{equation}
The probability of the rooted gene tree $T^+$ is then the  sum of the probabilities of all coalescent histories for $T^+$. Subsequent work built on this framework for faster computation \cite{Wu2012,Wu2016}.

For species networks, however, this notion of a coalescent history is inadequate, since when lineages reach a hybrid node, they may trace back through several different hybrid edges. This issue was circumvented in the algorithm of \cite{YuDegnanNakhleh2012} by first converting the network to a collection of multilabeled trees.
For the special case of quartet CFs, a fast recursive algorithm is presented in \cite{AnomQuartets} and implemented in the Julia package
\texttt{QuartetNetworkGoodnessFit}
\cite{QuartetNGF}.

All of the preceding were implemented for numerical computations, and for our purposes symbolic formulas are needed. While recently the Julia package has been adapted for symbolic quartet CF computation \cite{Kong2026}, it does not easily extend to quintets or beyond.  
 Thus we present a new approach that, in principal, computes symbolic gene tree probabilities for arbitrary phylogenetic networks.

\subsection{Recursive algorithm} 
For a recursive algorithm, we consider a partial coalescent process leading towards the formation of $T$ only in internal edges of $\cn$ near leaves. 
Then both the network $\cn$ and the gene tree $T$ are modified, allowing for recursive computation on smaller graphs.
Full coalescent histories are never made explicit in this approach.  

We formulate the algorithm for unrooted gene trees $T$, dealing with all possible rootings as it proceeds.
Here $T$ may be an unrooted binary topological tree on any subset of the taxa on $\cn$, so by enumerating
all $T$ on subsets of $m$ taxa, all $m$-tet CFs can be found.

With taxa $X$ on $\cn$, let $Y\subseteq X$ be the taxa on $T$.
First, we reduce to the case $X = Y$.
Indeed, an edge $e = (u,v) \in E(\cn)$ where $\mathrm{desc}_X(v) \subseteq X \setminus Y$ cannot
have a coalescent event involving the taxa $Y$, so
all such edges can be deleted from $\cn$.

To understand the recursion, note that
if a network $(\cn,\theta)$ is not a star tree it will have some edge, $e=(u,v)$, all of whose child edges are pendant. 
Choosing such an edge, we consider two cases: $e$ is a tree edge, and $e$ is a hybrid edge. 

If $e=(u,v)$ is a tree edge, then no coalescent event can occur below $v$ and the events involving only $Y$ lineages that occur on $e$ will result in the formation of a gene forest on $Y\cap \desc_X(v)$. 
These are rooted binary forests which are compatible with $T$, in the following sense:
A forest $F$ is \emph{compatible} with $T$ if for some rooting  $T^+$ of $T$, its elements are obtained as disjoint subtrees  below some nodes of $T^+$. 
For convenience, we label the roots of these subtrees by the set of taxa on them.
We denote the set of such forests by
$\cf(\cn,T,e).$
An example of such a set is given in \Cref{fig:coalescent-cases-uv}.

\begin{figure}[ht]
\centering
\begin{tikzpicture}[
    >=Stealth,
    every node/.style={font=\small},
    dot/.style={circle,fill=black,inner sep=1.3pt},
    leaf/.style={dot, label=below:{\small #1}},
    int/.style={dot},
]

\begin{scope}[shift={(1,2.2)}]
  \def\ovalA{1.4}   
  \def\ovalB{0.45}  
  \draw[dashed, line width=0.9pt]
  (0,{0.8+\ovalB}) ellipse [x radius=\ovalA, y radius=\ovalB];
  \node[int,label=above:$u$] (u) at (0,0.8) {};
  \node[int,label=right:$v$] (v) at (0,0.0) {};
  \draw[->,line width=0.8pt] (u) -- (v);

  \node[leaf=$A$] (A) at (-0.7,-0.9) {};
  \node[leaf=$B$] (B) at (0,-0.9) {};
  \node[leaf=$C$] (C) at (0.7,-0.9) {};

  \draw[->,line width=0.8pt] (v) -- (A);
  \draw[->,line width=0.8pt] (v) -- (B);
  \draw[->,line width=0.8pt] (v) -- (C);

  \node[anchor=east] at (-1.55,0.35) {$\cn \;=$};
  \node[below] at (0, -1.6) {(a)};
\end{scope}

\begin{scope}[shift={(5,2.2)}]
  \node[anchor=east] at (-1.1,0.35) {$T \;=$};

  \node[int] (x) at (0,0) {};
  \node[int] (y) at (1.8,0) {};
  \node[int] (z) at (0.9,0) {};

  \node[dot,label=below:$A$] (TA) at (-0.8,-0.8) {};
  \node[dot,label=above:$B$] (TB) at (-0.8, 0.8) {};
  \node[dot,label=above:$C$] (TC) at (0.9, 1.0) {};
  \node[dot,label=above:$D$] (TD) at (2.6,0.8) {};
  \node[dot,label=below:$E$] (TE) at (2.6,-0.8) {};

  \draw[line width=0.8pt] (x) -- (y);
  \draw[line width=0.8pt] (x) -- (TA);
  \draw[line width=0.8pt] (x) -- (TB);
  \draw[line width=0.8pt] (z) -- (TC);
  \draw[line width=0.8pt] (y) -- (TD);
  \draw[line width=0.8pt] (y) -- (TE);.

  \node[below] at (0.9,-1.6) {(b)};
\end{scope}

\def\topY{0.8}
\def\botY{-1.0}
\def\pipeextra{0.25}

\newcommand{\pipes}{%
  \draw[line width=1.2pt] (-0.35,\botY-\pipeextra) -- (-0.35,\topY+\pipeextra);
  \draw[line width=1.2pt] (1.95,\botY-\pipeextra) -- (1.95,\topY+\pipeextra);
}

\begin{scope}[shift={(-1,-1.5)}]
  \pipes
  \foreach \x/\lab in {0/$A$, 0.8/$B$, 1.6/$C$} {
    \node[dot,label={below:{\small \lab}}] at (\x,{(\botY + \topY)/2}) {};
  }
  \node[below] at (0.8, \botY-0.6) {(i) $F_1$};

\end{scope}

\begin{scope}[shift={(3,-1.5)}]
  \pipes
  \node[dot,label=below:{\small $A$}] (a2) at (0,{(2*\botY + \topY)/3}) {};
  \node[dot,label=below:{\small $B$}] (b2) at (0.8,{(2*\botY + \topY)/3}) {};
  \node[dot,label=below:{\small $C$}] (c2) at (1.6,{(2*\botY + \topY)/3}) {};

  \coordinate (m2) at (0.4,{(\botY + 2*\topY)/3});
  \draw[line width=0.8pt] (a2) -- (m2) node[dot] {} node[above] {$AB$};
  \draw[line width=0.8pt] (b2) -- (m2);

  \node[below] at (0.8, \botY-0.6) {(ii) $F_2$};

\end{scope}

\begin{scope}[shift={(7,-1.5)}]
  \pipes
  \node[dot,label=below:{\small $A$}] (a3) at (0,{(3*\botY+\topY)/4}) {};
  \node[dot,label=below:{\small $B$}] (b3) at (0.8,{(3*\botY+\topY)/4}) {};
  \node[dot,label=below:{\small $C$}] (c3) at (1.6,{(3*\botY+\topY)/4}) {};

  \coordinate (m31) at (0.4,{(\botY + \topY)/2});
  \coordinate (m32) at (0.95,{(\botY+3*\topY)/4});

  \draw[line width=0.8pt] (a3) -- (m31) node[dot] {};
  \draw[line width=0.8pt] (b3) -- (m31);
  \draw[line width=0.8pt] (m31) -- (m32) node[dot] {} node[above] {$ABC$};

  \draw[line width=0.8pt] (c3) -- (m32);
  \node[below] at (0.8, \botY-0.6) {(iii) $F_3$};
\end{scope}

\end{tikzpicture}
\caption{
For case 1 of the recursion, a tree edge $(u,v)$ of the current network $\cn$ is chosen with only pendant edges below $v$, here leading to $A,B,C$. Three forests $F_i$ on $A,B,C$ for this edge can be obtained by rooting the gene tree $T$ on some edge, say the one leading to $E$, to obtain $T^+$, and then taking disjoint subtrees below certain of its nodes. Note that two additional forests on $A,B,C$, given by exchanging the label $C$ with $A$ or $B$ in $F_2$ are not compatible with $T$, and therefore not shown.
}
\label{fig:coalescent-cases-uv}
\end{figure}
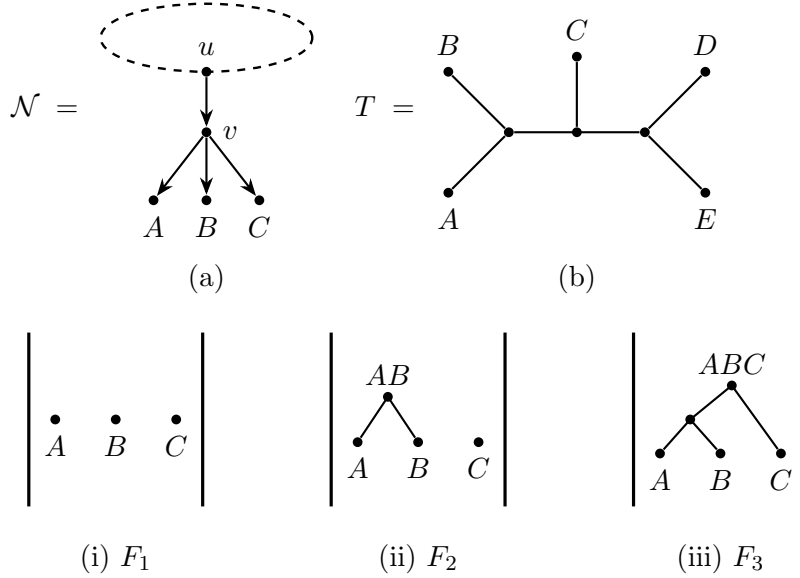

For such a forest $F \in \cf(\cn,T,e)$ with $i$ taxa and $j$ roots, the probability of $F$ forming in this edge is given by
\begin{equation}
\label{eqn:prob-in-edge}
    \mathbb P(F,e,x_e) = \frac{c(F,e)}{\prod_{k=j+1}^i \binom{k}{2}}\cdot g_{ij}(x_e)
\end{equation}
where $c(F,e)$ is the number of ordered coalescent histories within $e$ giving rise to $F$. This formula is justified as follows. 
The term $g_{ij}(x_e)$ gives the conditional  probability that $j$ distinct lineages reach $u$ given that $i$ lineages enter $e$ at $v$. Since each ordering of coalescent events in $e$ is equally likely,
the first factor gives the proportion of them that give rise to $F$. 
The product is thus the probability of $F$ forming on this edge. 

\begin{remark} 
    The function $c(F,e)$ is well-studied and can be computed using Knuth's hook-length formula for forests \cite{KnuthArt3}. Ancestry in the forest $F$ determines a partial order on its internal nodes, and hence on the coalescent events. The hook length formula counts extensions of this partial order to linear ones. 
\end{remark}

We next modify $\cn$ and $T$ according to the forest $F$.
Let $\cn_F$ be the network obtained from $\cn$ by deleting the edge $e$ and all its descendants, and for each tree $t$ in $F$ attaching a new edge with parent $u$ and child labeled by $t$'s root label. $\cn_F$ inherits parameters $\theta_F$ from those of $\cn$, since lengths of pendant edges need not be specified under the NMSC with one lineage per taxon.
To obtain $T_F$, for each tree in $F$ we delete all but its root node in $T$, labeling it as it is in $F$. See \Cref{fig:NFTF} for an example.

The first case of the recursion is now straightforward to state. If $e = (u,v)$ is a tree edge of $\cn$ with only pendant edges below $v$, then
\begin{equation}
\label{eqn:regular-recursion}
    \mathbb{P}(T~|~(\cn,\theta)) = \sum_{F\in\cf(\cn,T,e)} \mathbb P (F,e,x_e) \cdot \mathbb{P}(T_F~|~(\cn_F,\theta_F)).
\end{equation}
This simply expresses that $T$ may only form through  certain partial coalescent trees forming in $e$, together with other constrained coalescent events in the rest of $\cn$. 

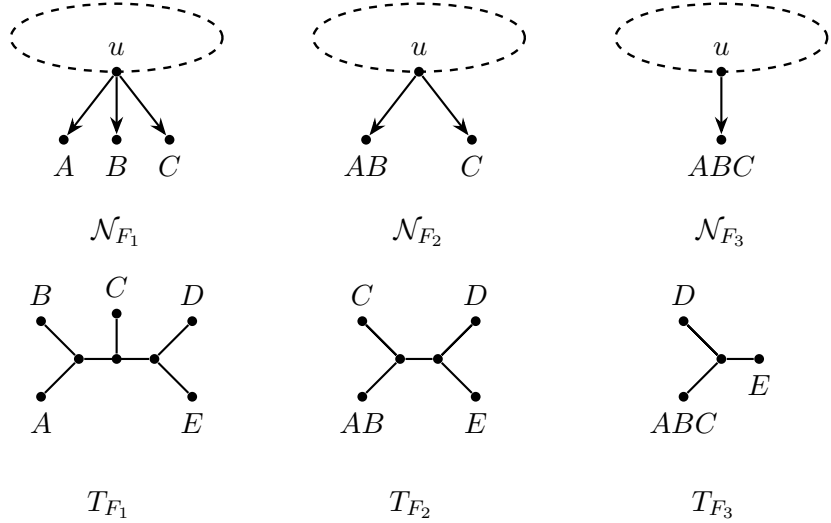
\begin{figure}[ht]
\begin{center}
    \begin{tikzpicture}[
        >=Stealth,
        every node/.style={font=\small},
        dot/.style={circle,fill=black,inner sep=1.3pt},
        leaf/.style={dot, label=below:{\small #1}},
        int/.style={dot},
    ]

    \begin{scope}[shift={(0,0)}]
        \def\ovalA{1.4}   
        \def\ovalB{0.45}  
        \draw[dashed, line width=0.9pt]
        (0,{0.8+\ovalB}) ellipse [x radius=\ovalA, y radius=\ovalB];
        \node[int,label=above:$u$] (uv) at (0,0.8) {};

        \node[leaf=$A$] (A) at (-0.7,-0.1) {};
        \node[leaf=$B$] (B) at (0,-0.1) {};
        \node[leaf=$C$] (C) at (0.7,-0.1) {};

        \draw[->,line width=0.8pt] (uv) -- (A);
        \draw[->,line width=0.8pt] (uv) -- (B);
        \draw[->,line width=0.8pt] (uv) -- (C);

        \node[below] at (0, -1) {$\cn_{F_1}$};
    \end{scope}

    \begin{scope}[shift={(4,0)}]
        \def\ovalA{1.4}   
        \def\ovalB{0.45}  
        \draw[dashed, line width=0.9pt]
        (0,{0.8+\ovalB}) ellipse [x radius=\ovalA, y radius=\ovalB];
        \node[int,label=above:$u$] (uv) at (0,0.8) {};

        \node[leaf=$AB$] (AB) at (-0.7,-0.1) {};
        \node[leaf=$C$] (C) at (0.7,-0.1) {};

        \draw[->,line width=0.8pt] (uv) -- (AB);
        \draw[->,line width=0.8pt] (uv) -- (C);

        \node[below] at (0, -1) {$\cn_{F_2}$};
    \end{scope}

    \begin{scope}[shift={(8,0)}]
        \def\ovalA{1.4}   
        \def\ovalB{0.45}  
        \draw[dashed, line width=0.9pt]
        (0,{0.8+\ovalB}) ellipse [x radius=\ovalA, y radius=\ovalB];
        \node[int,label=above:$u$] (uv) at (0,0.8) {};

        \node[leaf=$ABC$] (ABC) at (0,-0.1) {};

        \draw[->,line width=0.8pt] (uv) -- (ABC);

        \node[below] at (0, -1) {$\cn_{F_3}$};
    \end{scope}

    \begin{scope}[shift={(-1,-3)}]

        \node[int] (x) at (0.5,0) {};
        \node[int] (y) at (1.0,0) {};
        \node[int] (z) at (1.5,0) {};

        \node[dot,label=below:$A$] (TA) at (0,-0.5) {};
        \node[dot,label=above:$B$] (TB) at (0, 0.5) {};
        \node[dot,label=above:$C$] (TC) at (1, 0.6) {};
        \node[dot,label=above:$D$] (TD) at (2,0.5) {};
        \node[dot,label=below:$E$] (TE) at (2,-0.5) {};

        \draw[line width=0.8pt] (x) -- (z);
        \draw[line width=0.8pt] (x) -- (TA);
        \draw[line width=0.8pt] (x) -- (TB);
        \draw[line width=0.8pt] (y) -- (TC);
        \draw[line width=0.8pt] (z) -- (TD);
        \draw[line width=0.8pt] (z) -- (TE);.

        \node[below] at (0.9,-1.6) {$T_{F_1}$};
    \end{scope}

    \begin{scope}[shift={(3,-3)}]
        \node[int] (x) at (.75,0) {};
        \node[int] (y) at (1.25,0) {};

        \node[dot,label=below:$AB$] (TAB) at (.25,-.5) {};
        \node[dot,label=above:$C$] (TC) at (.25, .5) {};
        \node[dot,label=above:$D$] (TD) at (1.75,.5) {};
        \node[dot,label=below:$E$] (TE) at (1.75,-.5) {};

        \draw[line width=0.8pt] (TAB) -- (x) -- (TC) -- (x) -- (y) -- (TD) -- (y) -- (TE);
        \node[below] at (0.9,-1.6) {$T_{F_2}$};
    \end{scope}

    \begin{scope}[shift={(7,-3)}]
        \node[int] (x) at (1,0) {};

        \node[dot,label=below:$ABC$] (TABC) at (.5,-.5) {};
        \node[dot,label=above:$D$] (TD) at (.5,.5) {};
        \node[dot,label=below:$E$] (TE) at (1.5,0) {};

        \draw[line width=0.8pt] (TABC) -- (x) -- (TD) -- (x) -- (TE);
        \node[below] at (0.9,-1.6) {$T_{F_3}$};
    \end{scope}

    \end{tikzpicture}
\end{center}
\caption{For the first case of the recursion depicted in \Cref{fig:coalescent-cases-uv}, each of the forests $F_1,F_2,F_3$ leads to a modified network $\cn_{F_i}$ and gene tree $T_{F_i}$.
}\label{fig:NFTF}
\end{figure}

For the second case of the recursion, suppose $e$ is a hybrid edge all of whose descendant edges are pendant. 
Specifically, assume $e=(u_1,v)$  and $(u_2,v)$ is the only other hybrid edge with child $v$, with 
hybridization parameters
$\lambda_1, \lambda_2$ respectively, where $\lambda_1+\lambda_2=1$.
If the taxa descended from $v$ are $y_1,\dotsc,y_k$, 
then in a realization of the coalescent process let $\cl_i\subseteq [k]$ be the indices of the lineages entering $(u_i,v)$ so $[k]=\cl_1\sqcup \cl_2$.
Then the probability of this partition is
$$\mathbb P(\cl_1,\cl_2,\lambda_1,\lambda_2)=
    \lambda_1^{|\cl_1|} \lambda_2^{|\cl_2|}.
$$
In this situation, define a network $\cn_{\cl_1,\cl_2}$ by deleting all nodes and edges below $u_1,u_2$,  introducing  two new vertices $v_1,v_2$ and directed edges $(u_1,v_1),(u_2,v_2)$, and then edges $(v_1, y_i)$ for $i \in \cl_1$ and $(v_2,y_j)$ for $j \in \cl_2$. If $\cl_i = \emptyset$, then we do not introduce the vertex $v_i$ as it will have no descendants.
The new network is given parameters $\theta_{\cl_1,\cl_2}$ inherited from $\theta$,  with the modification that the new edges $(u_i,v_i)$ have the length of $(u_i,v)$ in $\theta$.
This essentially `splits' the hybridization, with the hybrid edges becoming tree edges, as depicted in \Cref{fig:hybsplit}.

Then we have the following recursive formula.
\begin{equation}
\label{eqn:hybrid-recursion}
    \mathbb{P}(T~|~(\cn,\theta)) = \sum_{\cl_1\sqcup\cl_2 = [k]}  \mathbb P(\cl_1,\cl_2,\lambda_1,\lambda_2) \cdot \mathbb{P}(T~|~(\cn_{\cl_1,\cl_2},\theta_{\cl_1,\cl_2}))
\end{equation}

\begin{figure}[ht]

\begin{center}
    \begin{tikzpicture}[
        >=Stealth,
        every node/.style={font=\small},
        dot/.style={circle,fill=black,inner sep=1.3pt},
        leaf/.style={dot, label=below:{\small #1}},
        int/.style={dot},
    ]

    \begin{scope}[shift={(4,0)}]
        \def\ovalA{1.4}   
        \def\ovalB{0.45}  
        \draw[dashed, line width=0.9pt]
        (0,{0.8+\ovalB}) ellipse [x radius=\ovalA, y radius=\ovalB];

        \node[int,label=above:$u_1$] (u1) at (-0.7,0.86) {};
        \node[int,label=above:$u_2$] (u2) at (0.7,0.86) {};
        \node[int,label=right:$v$] (v) at (0,0) {};
        \node[leaf=$y_1$] (A) at (-0.7,-0.86) {};
        \node[leaf=$y_2$] (B) at (0,-0.86) {};
        \node[leaf=$y_3$] (C) at (0.7,-0.86) {};

        \draw[->,line width=0.8pt] (u1) -- (v);
        \draw[->,line width=0.8pt] (u2) -- (v);
        \draw[->,line width=0.8pt] (v) -- (A);
        \draw[->,line width=0.8pt] (v) -- (B);
        \draw[->,line width=0.8pt] (v) -- (C);

        \node[below] at (0, -1.5) {$\cn$};
    \end{scope}

    \begin{scope}[shift={(0,-4)}]
        \def\ovalA{1.4}   
        \def\ovalB{0.45}  
        \draw[dashed, line width=0.9pt]
        (0,{0.8+\ovalB}) ellipse [x radius=\ovalA, y radius=\ovalB];

        \node[int,label=above:$u_1$] (u1) at (-0.7,0.86) {};
        \node[int,label=above:$u_2$] (u2) at (0.7,0.86) {};
        \node[int,label=left:$v_1$] (v1) at (-0.7,0) {};
        \node[int,label=right:$v_2$] (v2) at (0.7,0) {};
        \node[leaf=$y_1$] (A) at (-1,-0.86) {};
        \node[leaf=$y_2$] (B) at (-.4,-0.86) {};
        \node[leaf=$y_3$] (C) at (0.7,-0.86) {};

        \draw[->,line width=0.8pt] (u1) -- (v1);
        \draw[->,line width=0.8pt] (u2) -- (v2);
        \draw[->,line width=0.8pt] (v1) -- (A);
        \draw[->,line width=0.8pt] (v1) -- (B);
        \draw[->,line width=0.8pt] (v2) -- (C);
        \node[below] at (0, -1.5) {$\cn_{\{1,2\},\{3\}}$};
    \end{scope}

    \begin{scope}[shift={(4,-4)}]
        \def\ovalA{1.4}   
        \def\ovalB{0.45}  
        \draw[dashed, line width=0.9pt]
        (0,{0.8+\ovalB}) ellipse [x radius=\ovalA, y radius=\ovalB];

        \node[int,label=above:$u_1$] (u1) at (-0.7,0.86) {};
        \node[int,label=above:$u_2$] (u2) at (0.7,0.86) {};
        \node[int,label=left:$v_1$] (v1) at (-0.7,0) {};
        \node[int,label=right:$v_2$] (v2) at (0.7,0) {};
        \node[leaf=$y_1$] (A) at (0.4,-0.86) {};
        \node[leaf=$y_2$] (B) at (-0.7,-0.86) {};
        \node[leaf=$y_3$] (C) at (1,-0.86) {};

        \draw[->,line width=0.8pt] (u1) -- (v1);
        \draw[->,line width=0.8pt] (u2) -- (v2);
        \draw[->,line width=0.8pt] (v1) -- (B);
        \draw[->,line width=0.8pt] (v2) -- (A);
        \draw[->,line width=0.8pt] (v2) -- (C);

        \node[below] at (0, -1.5) {$\cn_{\{2\},\{1,3\}}$};
    \end{scope}

    \begin{scope}[shift={(8,-4)}]
        \def\ovalA{1.4}   
        \def\ovalB{0.45}  
        \draw[dashed, line width=0.9pt]
        (0,{0.8+\ovalB}) ellipse [x radius=\ovalA, y radius=\ovalB];

        \node[int,label=above:$u_1$] (u1) at (-0.7,0.86) {};
        \node[int,label=above:$u_2$] (u2) at (0.7,0.86) {};
        \node[int,label=left:$v_1$] (v1) at (-0.7,0) {};
        \node[leaf=$y_1$] (A) at (-1.2,-0.86) {};
        \node[leaf=$y_2$] (B) at (-0.7,-0.86) {};
        \node[leaf=$y_3$] (C) at (-0.2,-0.86) {};

        \draw[->,line width=0.8pt] (u1) -- (v1);
        \draw[->,line width=0.8pt] (v1) -- (B);
        \draw[->,line width=0.8pt] (v1) -- (A);
        \draw[->,line width=0.8pt] (v1) -- (C);

        \node[below] at (0, -1.5) {$\cn_{\{1,2,3\},\emptyset}$};
    \end{scope}
    \end{tikzpicture}
\end{center}

    \caption{For the second case of the recursion, with $e_1=(u_1,v)$ and $e_2=(u_2,v)$ hybrid edges of the current  network $\cn$, we `split' the hybrid node $v$ for each bipartition of its descendant taxa. Here 3 of the 8 such splittings are shown. Note that in the third splitting, $\cl_1 = \{1,2,3\}$, so we do not add $v_2$ to the graph.}\label{fig:hybsplit}
\end{figure}
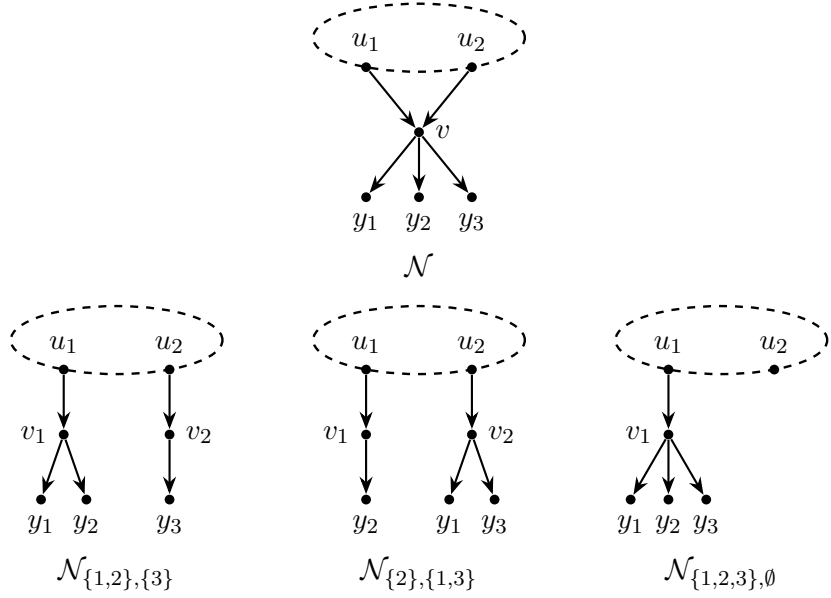

\begin{remark}
    Although we assumed the hybrid node $v$ had exactly two parents, this is sufficient even for a non-binary network.  If there are more parents of the hybrid node, additional hybridizations  with edges of length 0 can be introduced to reach this case.
\end{remark}

Repeated applications of equations \Cref{eqn:regular-recursion} and \Cref{eqn:hybrid-recursion}, will either reduce the number of hybrid nodes (increasing the tree edges by 2) or reduce the number of tree edges, and thus will  eventually modify the network to a star tree, i.e., with all edges pendant. 
This corresponds to all lineages reaching the root, and entering the infinite length above-the-root edge where they must eventually coalesce.
This requires a final probability calculation, which is essentially a special case  of equation \Cref{eqn:prob-in-edge}.

\begin{lemma}
\label{lemma:prob-of-tree-at-root}
    Suppose $k$ labeled lineages enter a single population of infinite duration under the coalescent model, and let $T$ be an unrooted topological gene tree on these labels. Then
    \begin{equation}
    \label{eqn:prob-of-tree-at-root}
        \mathbb P(T) = \frac{1}{\prod_{i=2}^{k}\binom{i}{2}}\sum_{T^+} c(T^+)
    \end{equation}
    where the sum is over all $2k-3$ choices of edges on which to place the root of $T$, and $c(T^+)$ is the number of ordered coalescent histories giving rise to $T^+$.
\end{lemma}

While included in the last result, since there is only one unrooted tree topology on three taxa, it is immediate that if $k=3$, then $\mathbb{P}(T~|~\cn) = 1$.

\begin{remark}
    If $k \leq 5$, the distribution on gene trees appearing in \Cref{lemma:prob-of-tree-at-root} is uniform. Indeed, the probability of any such gene tree occuring is $1/(2k-5)!!$. This is no longer true if $k > 5$. For example if $k = 6$, there are two distinct unlabeled gene tree topologies which can be classified by the number of cherries. There are 15 gene trees with three cherries and 90 with two cherries. Using \Cref{eqn:prob-of-tree-at-root}, the probability of observing any specific gene tree with three cherries is $1/75$ and $2/225$ for a gene tree with two cherries.
\end{remark}

We present the full recursive probability calculation for a gene tree as \Cref{alg:SCM}. 
and implemented this algorithm in \texttt{Macaulay2} \cite{M2}. We use transformed parameters, with edge probabilities instead 
of edge lengths, so each $g_{ij}$ is polynomial, so the final formula is as well. 
Code is available at GitHub (\href{https://github.com/jcu237/SymbolicCoalescentModel}{https://github.com/jcu237/SymbolicCoalescentModel}).

\begin{remark}
Computation time of this algorithm is of course affected by number of taxa on $\cn$, but even for a species tree the number of edges between root and leaves affects the growth of the forest sets, and hence the number of recursive calls. The number of hybrid nodes, and how many descendants each has is also important. Since there is a positive probability that no coalescent events occur below a hybrid node, if it has $k$ descendant taxa, a single forest at one may generate $2^k$ recursive calls.
\end{remark}

\begin{algorithm}[H]
\caption{Symbolic Coalescent Model}\label{alg:SCM}
\KwIn{a metric rooted network $(\cn,\theta)$ on $X$, $\theta=(\{x_e\},\{\lambda_e\})$; an unrooted binary topological gene tree $T$ on  $Y\subseteq X$}

\KwOut{SCM$(\cn,T,\theta) = \mathbb{P}(T~|~(\cn,\theta))$ under the NMSC model}
\ 

\If{$T$ has $\leq 3$ leaves}{
  \Return $1$\;
}\
\ElseIf{there exists an edge $e = (u,v)$ of $\cn$ with only pendant edge descendants}{
\If{$e$ is a tree edge}{ 
  For each $F \in \cf(\cn,T,e)$, compute $\mathbb P(F,e,x_e)$\;
  \Return $\sum_{F\in\cf} \mathbb P (F,e,x_e) \cdot \mathrm{SCM}(\cn_F,T_F,\theta_F)$\;
}\
\Else{
  $e$ is a hybrid edge;\
  Let $e_1,e_2$ be the parental edges of $v$; $\lambda_1,\lambda_2$ their hybridization parameters\;
  $\{y_i\}_{i\in[k]}$ the descendants of $v$\;
  \Return $\sum_{\cl_1\sqcup\cl_2=[k]} \mathbb P(\cl_1,\cl_2,\lambda_1,\lambda_2) \cdot \mathrm{SCM}(\cn_{\cl_1,\cl_2}, T_{\cl_1,\cl_2},\theta_{\cl_1,\cl_2})$\;
}\
}\
\Else{$\cn$ is a star network with $k > 3$ leaves\;
  \Return $\frac{1}{\prod_{i=2}^{k}\binom{i}{2}}\sum_{T^+} c(T^+)$\;
}
\end{algorithm}

\begin{example}
\label{example: SCM_exa}
    Consider the level-2 network $\cn$ from \Cref{fig1}. 
   The file \texttt{exa3\_6.m2} in our GitHub \href{https://github.com/jcu237/SymbolicCoalescentModel}{repository},
    gives code that computes the probability of observing
    each of the 15 gene trees on $\cn$.
    This took 345.9 seconds on an HP Elitebook with an Intel core ultra 7 chip.
    For example, the probability of observing $T_{15}$ from \Cref{tab:gene_tree_ordering} of \Cref{app:TandF} is
    \begin{align*}
        \Pp(T_{15} ~|~ (\cn, (\{x_i\},\{\lambda_i\}))) &= \frac{1}{15}x_{1}^{3}x_{2}^{3}x_{4}x_{8}\lambda_1^{2}\lambda_2-\frac{2}{15}x_{1}^{3}x_{2}^{3}x_{4}x_{8}\lambda_1\lambda_2+\frac{1}{15}x_{1}^{3}x_{2}^{3}x_{6}x_{8}\lambda_1\lambda_2\\
        &-\frac{1}{15}x_{1}^{3}x_{2}^{3}x_{6}x_{8}\lambda_1+\frac{1}{15}x_{1}^{3}x_{2}^{3}x_{4}x_{8}\lambda_2-\frac{1}{15}x_{1}^{3}x_{2}^{3}x_{6}x_{8}\lambda_2\\
        &+\frac{1}{15}x_{1}^{3}x_{2}^{3}x_{6}x_{8}-\frac{1}{15}x_{1}^{3}x_{2}x_{6}x_{8}\lambda_1\lambda_2-\frac{2}{15}x_{1}^{3}x_{2}x_{8}\lambda_1^{2}\lambda_2\\
        &+\frac{1}{15}x_{1}^{3}x_{3}x_{8}\lambda_1^{2}\lambda_2+\frac{1}{15}x_{1}^{3}x_{2}x_{6}x_{8}\lambda_1+\frac{2}{15}x_{1}^{3}x_{2}x_{8}\lambda_1\lambda_2
    \end{align*}
    where $x_i = \exp(\ell_i)$. Our code is flexible in that
    it is able to compute gene tree probabilities symbolically
    or numerically.
    We can either substitute in values for the edge lengths or rerun the code with floats as parameters. For example, for parameters given by
    \begin{align*}
        (x_1,\dotsc,x_8) &= (0.248559, 0.674666, 0.239703, 0.907135,\\
        &\;\;\;\;\;\;0.661391, 0.714215, 0.236131, 0.965198), \\
        \lambda_1 &= 0.659655, \\
        \lambda_2 &= 0.959008,
    \end{align*}
    the full CF vector was computed in 342.2 seconds to be
    \begin{align*}
        (u_1,\dotsc,u_{15}) &= (0.4099358, 0.0444048, 0.0444048, 0.192466, 0.0240773,  \\
        &\;\;\;\;\;\; 0.0240773, 0.00043239, 0.00043239, 0.0254128, 0.00043239,\\
        &\;\;\;\;\;\; 0.00043239, 0.0254128, 0.207793, 0.00043239, 0.00043239)
    \end{align*}
    where $u_i = \Pp(T_i ~|~ \cn)$.
\end{example}

\section{Application to Level-1 Species Networks}
\label{sec: root identifiability}

In this section,  
 gene tree probabilities computed by \Cref{alg:SCM} are used to study algebraic varieties arising from quintet CFs for binary level-$1$ species networks. For this study, we consider only 5-taxon networks, leaving what quintet CFs on larger networks imply as future work.
We further restrict our attention to 5-taxon level-$1$ networks  with no 2-cycles.
Doing so means we consider a finite number of 5-taxon networks, while allowing 2-cycles, which can be introduced repeatedly along cut edges, would give an infinite number. A full analysis of 5-taxon networks with 2-cycles will appear in a forthcoming work.

\subsection{Topological identifiability} 
Since quartet CFs can be obtained by marginalizing quintet CFs over one of their taxa, anything identifiable from quartet CFs is also identifiable from quintet CFs. Quartet concordance factors determine both the tree of blobs \cite{TreeOfBlobs2022} and a cyclic ordering of taxon groups around each blob for outer-labeled planar networks such as level-1 networks \cite{RhodesEtAl2024}. 
We may therefore assume that the underlying tree of blobs is one of the three shown in
\Cref{fig:blob-trees}, with the depicted planar embedding indicating the circular order. 

\begin{figure}[ht]
\centering
\begin{tikzpicture}[
    >=Stealth,
    every node/.style={font=\small},
    dot/.style={circle,fill=black,inner sep=1.3pt},
    leaf/.style={dot},
    int/.style={dot},
]

\begin{scope}[shift={(0,0)}]

  \node[int] (u) at (1,0) {};
  \node[int] (v) at (2,0) {};
  \node[int] (w) at (3,0) {};

  \node[leaf,label=below left:$A$] (A) at (0.3,-0.7) {};
  \node[leaf,label=above left:$B$] (B) at (0.3,0.7) {};
  \node[leaf,label=above:$C$] (C) at (2,1) {};
  \node[leaf,label=above right:$D$] (D) at (3.7,0.7) {};
  \node[leaf,label=below right:$E$] (E) at (3.7,-0.7) {};

  \draw[line width=0.8pt] (A) -- (u) -- (B) -- (u) -- (v) -- (C) -- (v) -- (w) -- (D) -- (w) -- (E);

  \node[below] at (2,-2) {(i)};
\end{scope}

\begin{scope}[shift={(5,0)}]
    \node[int] (u) at (1,0) {};
    \node[int] (v) at (2,0) {};
    \node[leaf,label=below:$A$] (A) at (1,-1) {};
    \node[leaf,label=left:$B$] (B) at (0,0) {};
    \node[leaf,label=above:$C$] (C) at (1,1) {};
    \node[leaf,label=above right:$D$] (D) at (2.7,0.7) {};
    \node[leaf,label=below right:$E$] (E) at (2.7,-0.7) {};

    \draw[line width=0.8pt] (A) -- (u) -- (B) -- (u) -- (C) -- (u) -- (v) -- (D) -- (v) -- (E);

    \node[below] at (1,-2) {(ii)};
\end{scope}

\begin{scope}[shift={(10,0)}]
    \node[int] (u) at (0,0) {};
    \node[leaf,label=right:$A$] (A) at (1,0) {};
    \node[leaf,label=above:$B$] (B) at (0.31,0.95) {};
    \node[leaf,label=above left:$C$] (C) at (-0.81,0.59) {};
    \node[leaf,label=below left:$D$] (D) at (-0.81,-0.59) {};
    \node[leaf,label=below:$E$] (E) at (0.31,-0.95) {};

    \draw[line width=0.8pt] (A) -- (u) -- (B) -- (u) -- (C) -- (u) -- (D) -- (u) -- (E);

    \node[below] at (0,-2) {(iii)};
\end{scope}
\end{tikzpicture}
\caption{The three trees of blobs, up to taxon labeling, on 5 taxa.}\label{fig:blob-trees}
\end{figure}
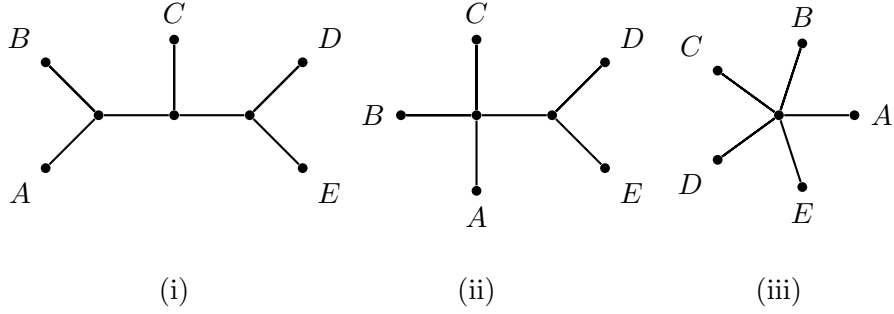

For each tree of blobs in \Cref{fig:blob-trees}, a $k$-cycle may be inserted at any internal vertex of degree $k$. 
In case (iii), since the hybrid node of a 5-cycle is identifiable from quartet CFs \cite{SNaQ, Banos2019} we assume it is the parent of $A$, with taxa arranged in alphabetical order, giving a single case of a  semidirected network to further analyze.

A 4-cycle in a 5-taxon network is also identifiable from quartet CFs \cite{AllmanEtAl2024}. However, we consider all cases for the location of the 4-cycles hybrid node in case (ii) of \Cref{fig:blob-trees} giving 4 cases of a semidirected network with no 3-cycles.

The degree 3 nodes of the tree of blobs in cases (i) and (ii) of \Cref{fig:blob-trees} may represent more topological types, 
as they could arise either from single nodes (trivial blobs) or from 3-cycles with two different choices of hybrid node depending on the location of the network root. 
Quartet CFs only allow these cases to be distinguished sometimes \cite[Theorem 19 (1,2)]{AllmanEtAl2024}. 
We therefore consider all ways in which the degree 3 nodes in cases (i) and (ii) may be either retained or replaced with 3 cycles (omitting cases that cannot be rooted), with all possible rootings.

There are, however, cases where the structure of a blob placed at a degree 3 node in a 5-taxon tree of blobs is not identifiable for certain rootings of the network. This follows from a more general result we present in  the following proposition. While it generalizes
a known case for quartets and level-1 networks \cite{AllmanEtAl2024}, it applies to arbitrary networks, and to arbitrary $m$-tets. 

\begin{proposition}\label{prop:cherries-non-identifiable}
    Suppose a rooted phylogenetic network $\cn$ contains a cherry with leaves $A$ and $B$ whose parent is node $v$. Let $\cn'$ be a network obtained from $\cn$ by deleting the cherry and identifying $v$ with the top node of any 3-blob with descendant edges to $A$ and $B$.
    (See \Cref{fig:3blobid} for an example.)
    Then for any fixed common choice of edge lengths and hybridization parameters not below $u$ for $\cn$ and $\cn'$, as parameters below $u$ on $\cn$ and $\cn'$ vary over all allowed stochastic values, the vectors of probabilities under the NMSC of all gene trees for the two networks 
    range over exactly the same sets.

    Thus, for any $m$, the $m$-tet CFs for these networks have exactly the same stochastic image, and no method using $m$-tet CFs can distinguish between $\cn$ and $\cn'$.
\end{proposition}
    
\begin{figure}
    \centering
    \begin{tikzpicture}[
        >=Stealth,
        every node/.style={font=\small},
        dot/.style={circle,fill=black,inner sep=1.3pt},
        leaf/.style={dot},
        int/.style={dot},
    ]
    \begin{scope}[shift={(0,0)}]
        \node[leaf,label=below:$A$] (A) at (0,0) {};
        \node[leaf,label=below:$B$] (B) at (2,0) {};
        \node[int, label=right:$v$] (v) at (1,1) {};
        \node[int, label=right:$u$] (u) at (1,2) {};
        \draw[->,line width=0.8pt] (u) -- (v);
        \draw[->,line width=0.8pt] (v) -- (A);
        \draw[->,line width=0.8pt] (v) -- (B);
    \end{scope}

    \begin{scope}[shift={(5,0)}]
        \node[leaf,label=below:$A$] (A) at (0,0) {};
        \node[leaf,label=below:$B$] (B) at (2,0) {};
        \node[int,label=left:$v_4$] (v1) at (0.5,0.5) {};
        \node[int,label=right:$v_3$] (v2) at (1.5,0.5) {};
        \node[int,label=right:$v_1$] (v3) at (1,1.5) {};
        \node[int,label=left:$v_2$] (v4) at (0.75,1) {};
        \node[int,label=right:$u$] (u) at (1,2) {};

        \draw[->,line width=0.8pt] (u) -- (v3);
        \draw[->,line width=0.8pt] (v3) -- (v2);
        \draw[->,line width=0.8pt] (v3) -- (v4);
        \draw[->,line width=0.8pt] (v4) -- (v1);
        \draw[->,line width=0.8pt] (v4) -- (v2);
        \draw[->,line width=0.8pt] (v2) -- (v1);
        \draw[->,line width=0.8pt] (v1) -- (A);
        \draw[->,line width=0.8pt] (v2) -- (B);
    \end{scope}
    \end{tikzpicture}
    \caption{On the left, a cherry with taxa $A$ and $B$ having parent node $v$, and on the right, a blob placed in this cherry at the location of $v$, as described in \Cref{prop:cherries-non-identifiable}} \label{fig:3blobid}
\end{figure}
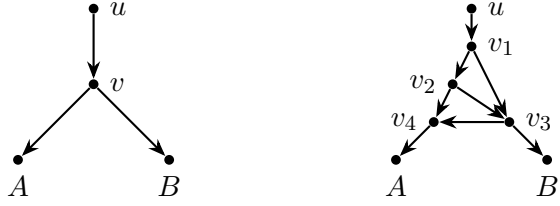

\begin{proof}
Consider any fixed 3-blob for $\cn'$. It must contain as a displayed rooted tree one in which $A,B$ form a cherry with a (composite) edge leading from their parent to $v$. By choosing hybridization parameters near 0 or 1 and edge lengths along this composite edge appropriately, we see the probability that the $A,B$ lineages do not coalesce in the 3-blob  ranges over all of  $(0,1)$ as all 3-blob parameters are varied. The probability that two lineages in the edge  $(u,v)$ of $\cn'$ do not coalesce also ranges over $(0,1)$ as the length is varied. 
The probability of non-coalescence below $u$ for the $A,B$ lineages is the product of these probabilities, and thus also ranges over $(0,1)$. But this is exactly the same as for $\cn$.
\end{proof}

Under the above assumptions and after removing networks with 3-cycles located on cherries of the rooted network as in \Cref{prop:cherries-non-identifiable},
there are $108$ rooted topological networks with tree of blobs (i), 
$73$ with tree of blobs (ii),
and $9$ with tree of blobs (iii)
for which we compute the quintet CF ideals $I_\cn$.

This count, and the specific networks, for tree of blobs (i) is obtained as follows. 
There are three degree-$3$ vertices where a $3$-cycle could be inserted: the vertex adjacent to the cherry ${A,B}$, the central vertex adjacent to $C$, and the vertex adjacent to the cherry ${D,E}$. 
By \Cref{prop:cherries-non-identifiable}, any $3$-cycle inserted at one of the cherry vertices can only be identifiable if the root lies on that $3$-cycle or one of the edges pendant to the $3$-cycle. 
From here on, we refer to the cherry edges or the collection of the five 3-cycle and pendant edges as an \emph{$AB$- or $DE$-cluster}.
Thus, if the network contains a single $3$-cycle, we may assume either that this $3$-cycle is placed at the central vertex, or that the root lies in the corresponding $3$-cycle cluster. 
The remaining cases are equivalent, at the level of concordance factors, to replacing the $3$-cycle by a trivial blob.

Similarly, if for tree of blobs (i) the network contains two $3$-cycles, then at least one must be placed at a cherry vertex.
If neither $3$-cycle is placed at the central vertex, then both occur at cherry vertices, but only the cluster
containing the root might be identifiable
by \Cref{prop:cherries-non-identifiable}, this reduces to a case with a single $3$-cycle. 
Therefore, the only genuinely new cases of two $3$-cycles are those in which one $3$-cycle is placed at the central vertex and the root lies in the other $3$-cycle cluster. 
Finally, networks with three $3$-cycles can be disregarded: the two $3$-cycles at cherry vertices cannot both contribute identifiable structures by \Cref{prop:cherries-non-identifiable}.

With these simplifications, taking into account the root location and hybrid node location in 3-cycles  for case (i) yields 7 networks with no 3-cycles, 49 with one 3-cycle and 52 with two 3-cycles.

For tree of blobs (ii), there is only one degree-$3$ vertex, namely the vertex adjacent to the cherry ${D,E}$. 
Consequently, we only need to consider a 3-cycle on networks of this type when the root lies in the $DE$-cluster; 
otherwise, the $3$-cycle is invisible to concordance factors and the network reduces to the corresponding case with a trivial blob at that vertex.
Accounting for root and hybrid node location gives 39 networks of type (ii) with a 3-cycle and 34 without.

The count for networks of type (iii), given that $A$ is the hybrid child is 9, since there are 9 edges on which the root might be located.

\subsection{Setup of computation}
Throughout, we fix an ordering of gene trees $T_1,\dotsc,T_{15}$ (see \Cref{tab:gene_tree_ordering} of \Cref{app:TandF}) with taxon set $\{A,B,C,D,E\}$, and for a network $\cn$, let $E(\cn)$ denote the edges of $\cn$ and $H(\cn)\subseteq E(\cn)$ denote the hybrid edges of $\cn$.

For each of the 190 rooted networks, we computed $u_i = \Pp(T_i ~|~ (\cn, (\{x_e\}, \{\lambda_f\})))$ symbolically via \Cref{alg:SCM}. These define a ring homomorphism
\begin{align*}
    \varphi_\cn : \Cc[u_1,\dotsc,u_{15}] &\to \Cc[x_e, \lambda_f ~|~ e \in E(\cn), f\in H(\cn) ], \\
    u_i &\mapsto \Pp(T_i ~|~ (\cn, (\{x_e\},\{\lambda_f\}))).
\end{align*}
We would like to compute the $\ker \varphi_\cn$ directly, but
for many of the networks under consideration, the necessary Gr{\"o}bner basis calculation
does not terminate on machines available to us.
Because $\varphi_\cn$ is not a graded map,
some techniques from computational algebra are not directly applicable.
To rectify this situation, we introduce
homogenizing variables: $u_0$ (to the domain) and $t$ (to the codomain).
We therefore consider the ring homomorphism.
\begin{align*}
    \varphi_\cn^{\text{hom}} : \Cc[u_0,\dotsc,u_{15}] &\to \Cc[t,x_e,\lambda_f ~|~ e \in E(\cn), f\in H(\cn)], \\
    u_0 &\mapsto t, \\
    u_i &\mapsto \varphi_\cn(u_i) \cdot t ,\hspace{1.5cm} 1\le i\le 15.
\end{align*}
Then $\varphi_\cn^{\text{hom}}$ is $\Zz$-graded with $\deg(u_i) = 1$, $\deg(t) = 1$, and $\deg(x_e) = \deg(\lambda_f) = 0$.
In particular, the $\ker \varphi_\cn^{\text{hom}}$ is $\Zz$-graded,
so the \texttt{Macaulay2} \cite{M2} package \texttt{MultigradedImplicitization} \cite{CummingsHollering2026} can be used to compute all elements of $\ker \varphi_\cn^{\text{hom}}$ of low degree.

\begin{remark}\label{rem:homogenization_is_okay}
    The ideals $\ker \varphi_\cn$ and $\ker \varphi_\cn^{\text{hom}}$ are closely related. From a generating set of $\ker \varphi_\cn^{\text{hom}}$, one can recover a generating set for $\ker \varphi_\cn$ by substituting $u_0 \mapsto 1$.
    Thus, no information is lost through homogenization.
\end{remark}

Using $\varphi_\cn^{\text{hom}}$ we computed a minimal generating set up to degree 4 for all $190$ networks, and sorted the networks by the results.
This gives a first pass at the distinguishability results presented in the next subsection.
But unlike a Gr{\"o}bner basis computation, these computations do not produce a certificate that the full ideal has been found, so
further investigation is required.

For each class of networks with the same ideal up to degree 4, we checked whether the corresponding ideal was prime and had the same dimension as the rank of the Jacobian matrix for the parameterizations of each network in the class. 
If the ideal is prime and the dimension matches all the ranks, this proves that the full kernel has been found,
and no further investigation of the class is needed.
In order to be self-contained, we include a sketch of the proof of this fact below, as \Cref{lem:primedim,lem:Jacdim}.

If, on the other hand, the dimension and ranks did not match or the ideal was not prime, then we needed to do more to determine whether networks in this particular class are distinguishable. 
Our first step was to use the linear invariants (which are the same for all networks in a class), to reduce the number of variables in the domain. 
In some cases, this reduction allowed Gr{\"o}bner computations to finish, while in others, it allowed us to compute higher degree invariants using \texttt{MultigradedImplicitization}.
In a few cases, both of these methods failed, but we were able to use algebraic matroids, as implemented in the \texttt{Matroids} package 
\cite{MatroidsSource, MatroidsArticle},
to distinguish between networks in the class \cite{HolleringSullivant2021}. 
After all these tests were done, we arrive at the distinguishability statements found in the theorems of the following subsection.

\begin{remark}
    Since the parameterizations of $\varphi_\cn$ only involve rational coefficients,
    computations done in a computer algebra system working over the rationals
    is sufficient for finding invariants for the full ideals defined over the complex  (see \Cref{lem:genset} below).
    In particular, say we find generators of $\ker \varphi_\cn$ up to some finite degree $d$, 
    and we let $J_d$ be the ideal generated by the invariants in $\ker \varphi_\cn$ up to degree $d$.
    Then if $J_d$ is prime in the polynomial ring with rational coefficients and $\dim J_d$ is equal to the rank of the Jacobian, then $J_d = \ker \varphi_\cn$.
    Importantly, these steps can often be carried out on a computer and constitute a certificate that a complete generating set has been found.
\end{remark}
    
\medskip

 The following lemmas justify the approach described above when Gr\"obner computations failed to terminate. 

\begin{lemma}\label{lem:primedim}
    Let $I \subseteq J$ be prime ideals in a Noetherian ring $R$. If $\mathrm{dim}(I) = \mathrm{dim}(J)$, then $I = J$. 
\end{lemma}
\begin{proof}
    Strict inclusion of prime ideals increase height \cite[Chapter~9]{EisenbudCA}, 
    so $I = J$.
\end{proof}

\begin{lemma}\label{lem:Jacdim}
    Let $k$ be a field of characteristic 0.
    Let $\varphi : k[x_1,\dotsc,x_n] \to k[\theta_1,\dotsc,\theta_m]$ be a $k$-algebra homomorphism, and let $\phi : k^m \to k^n$ be the associated polynomial map of the form
    \[
        \phi(\theta_1,\dotsc,\theta_m) = (\phi_1(\theta),\dotsc,\phi_n(\theta)).
    \]
    Then $\dim(\ker \varphi)$ is the rank of the Jacobian of $\phi$
    \[
        J(\phi) = \left(\frac{\partial \phi_j}{\partial \theta_i} \right), \; 1\leq i \leq m, \; 1\leq j\leq n
    \]
\end{lemma}
\begin{proof}
    Recall that the dimension of $\ker \varphi$ is the Krull dimension of the quotient ring
    \[
        k[x_1,\dotsc,x_n]/\ker \varphi.
    \]
    By the description of differentials of field extensions
    \cite[Theorem~16.14]{EisenbudCA}, the generic rank of the Jacobian
    of the parametrization is 
    \[
        \mathrm{trdeg}_k k(\phi_1,\dotsc,\phi_n),
    \]
    which equals the Krull dimension $\dim(k[x_1,\dotsc,x_n]/\ker\varphi)$
    by the dimension theorem for affine domains
    \cite[Theorem~A]{EisenbudCA}.
\end{proof}

\begin{lemma}\label{lem:genset}
    Let $\varphi_k : k[x_1,\dotsc,x_n] \to k[\theta_1,\dotsc,\theta_m]$ be a $k$-algebra homomorphism, let $K$ be a field extension, and let $I_k$ be the kernel of $\varphi_k$. 
    Consider the map $\varphi_K$ defined by extension of scalars.
    Then $\ker \varphi_K = I_k \otimes_k K$, i.e. a generating set of $I_k$ (lifted to the larger ring) is a generating set for $\ker \varphi_K$.
\end{lemma}

\begin{proof}
    Any field extension $K/k$ is a $k$-vector space. In particular, $K$ is free over $k$ and thus flat over $k$.  Thus tensoring the left exact sequence
    \[
        0 \to I_k \to k[x_1,\dotsc,x_m] \xrightarrow[]{\varphi_k} k[\theta_1,\dotsc,\theta_m]
    \]
    with $K$, gives the left exact sequence
    \[
        0 \to I_k \otimes_k K \to K[x_1,\dotsc,x_m] \xrightarrow[]{\varphi_K} K[\theta_1,\dotsc,\theta_m].
    \]
    It follows that $I_k \otimes_k K$ is equal to $I_K$.
\end{proof}

\subsection{Results}
It is important to emphasize that these results should not all be interpreted as sharp identifiability statements. 
Rather these computations provide proofs that certain classes of networks definitively produce concordance factors satisfying different algebraic relations. In this case the networks are \emph{algebraically distinguishable} from one another. However, for other classes there may be missing distinguishability statements. We will point out where we think this is most likely. 

Moreover, we are only considering algebraic relations among the concordance factor vectors; however, these are really semialgebraic objects meaning that we should really be working over $\Rr$ and considering inequalities.
It may be the case that even if the ideals associated to two networks are the same,
they differ once we restrict to the real paramater space $\Theta_\cn$ described in \Cref{sec: background}.

Our first theorem pertains to the tree of blobs (i) from \Cref{fig:blob-trees}.
For each network with tree of blobs (i),
the ideals generated by invariants up to degree four were prime (over the rationals) and had the expected dimension in every case examined. 
Consequently, these computations recover the full vanishing ideals, and hence that the corresponding algebraic identifiability statements are complete.

\begin{theorem}
    \label{thm:blob-tree-i} Let $\cn$ be a rooted level-1 binary network without 2-cycles whose tree of blobs is the fully resolved tree on five taxa pictured in \Cref{fig:blob-trees} (i). 
    Then the following statements of algebraic identifiability from quintet CFs for generic parameter values hold.

    \begin{enumerate}
        \item \textbf{Whether a $3$-cycle is present at the central vertex,}
         i.e.\ the vertex separating $C$ from the cherries $AB$ and $DE$ in the tree of blobs, can be determined.

        \item \textbf{If a $3$-cycle is present at the central vertex,} then
        whether or not other 3-cycles (at the AB and DE cherries) are present cannot be determined. One can determine if
        \begin{enumerate}
        \item the root lies on either of the internal cut edges adjacent to the central 3-cycle,
        \item the root lies on an edge of the central 3-cycle or the pendant edge to $C$ in the semidirected network, but cannot distinguish between these possibilities,
        \item the root lies in the $AB$-cluster (resp. the $DE$-cluster), i.e, on the pendant edges to the two taxa or on an adjacent $3$-cycle, but cannot distinguish between these possibilities.
        \end{enumerate}

        \item \textbf{If no $3$-cycle is present at the central vertex,}
        \begin{enumerate}
        \item If the root lies within the $AB$-cluster (resp.\ $DE$-cluster), then the presence or absence  of a $3$-cycle  adjacent to those taxa is determined. If the root does not lie in such a cluster, then the presence of a an additional $3$-cycle cannot be determined.
            \item The root location is identifiable except when it lies in the $AB$-cluster (resp. $DE$-cluster) and a $3$-cycle is present in the cluster. In this case the root can only be determined to be in the cluster.
        \end{enumerate}

        \item \textbf{Hybrid nodes within $3$-cycles are not identifiable} in all cases, except as constrained to two or three nodes by the location of the root.
    \end{enumerate}
\end{theorem}

\begin{proof}
    The proofs of the statements above are all entirely computational. 
    All invariants were computed up to degree 4 and sorted accordingly. 
    In all cases, the ideals were prime and their dimension was equal to the rank of the Jacobian; thus, we conclude that these are the full vanishing ideals.
    The necessary computations can be found on our \href{https://github.com/jcu237/SymbolicCoalescentModel}{GitHub repository}.
\end{proof}

In contrast to the completeness of the previous theorem, the theorem below should be viewed as a partial result derived from the currently computable invariants. 
In those cases, the ideals generated in degree at most four do not capture the entire algebraic structure of the models in many cases. 
We had to use algebraic matroids for distinguishing the root locations in parts (3) and (4) of \Cref{thm:blob-tree-ii}.
Moreover, in part (2) of \Cref{thm:blob-tree-ii}, we have numerical evidence that the root location can actually be identified; however, we do not have any certifiable computations which prove this fact.
Thus, the non-identifiability statements in this theorem should be interpreted cautiously.
They reflect indistinguishability with respect to the currently known low-degree invariants and matroidal information rather than definitive proofs of algebraic equivalence.

\begin{theorem}
\label{thm:blob-tree-ii}
Let $\cn$ be a rooted level-1 binary network without 2-cycles whose tree of blobs is the tree on five taxa  with a 4-multifurcation, with the 4-blob having circular order shown in \Cref{fig:blob-trees} (ii). 
 
Then the following statements of algebraic identifiability from quintet CFs hold.

\begin{enumerate}
\item \textbf{The hybrid node on the $4$-cycle} is identifiable.

\item \textbf{When the hybrid node in the $4$-cycle is ancestral to $D$ and $E$,} the
location of the root may be on the pendant edges to $A,B,$ or $C$, or on the $4$-cycle, and there may or may not be a 3-cycle in the $DE$-cluster. None of these cases can be distinguished.

\item \textbf{When the hybrid node in the $4$-cycle is ancestral to $A$ or $C$,}
then a 3-cycle in the $DE$-cluster is not detectable, unless the root is in this cluster.
The location of the root is identifiable with one exception: when the root is in the $DE$-cluster and there is a 3-cycle,  the root can only be determined to be in the $DE$-cluster.

\item \textbf{When the hybrid node in the $4$-cycle is ancestral to $B$,}
then a $3$-cycle in the $DE$-cluster is not detectable, unless the root is in this cluster.
If the root is not in the $DE$-cluster, it's location can only be determined to be in one of three regions:
\begin{enumerate}
    \item It may be on the edge pendant to $C$ or either of the two $4$-cycle edges adjacent to that edge.
    \item It may be on the edge pendant to $A$ or either of the two $4$-cycle edges adjacent to that edge.
    \item It may be on the cut-edge separating the $DE$-cluster from the $4$-cycle.
\end{enumerate}
If the root is in the $DE$-cluster and there is no $3$-cycle there, then the root location is identifiable; otherwise, if there is a $3$-cycle, then the root location is only known up to the $DE$-cluster.

\item \textbf{3-cycle identifiability.}
The presence or absence of a 3-cycle is only known if the root is in the $DE$-cluster. In all cases, if there is a 3-cycle present, the hybrid location cannot be determined, beyond constraints imposed by the root location.
\end{enumerate}
\end{theorem}

\begin{proof}
    Again the proof is entirely computational, 
    and all necessary computations can be found in our GitHub repository.
    The invariants up to degree 4 are enough to determine the location of the hybrid node in the 4-cycle.
    In case (3), there are 2 root locations on the 4-cycle
    where algebraic matroids are used to distinguish the varieties,
    and the other 2 root locations on the 4-cycle were distinguished by reducing the number of variables using the linear invariants and then computing full Gr{\"o}bner bases.
See \Cref{fig:blob-tree-ii-rootings} and \Cref{tab:table-blob-tree-ii-ideal-info} of \Cref{app:TandF} for details.
    
\end{proof}

\begin{figure}[ht]
\begin{tikzpicture}[>=Stealth, every node/.style={font=\small},dot/.style={circle,fill=black,inner sep=1.3pt},leaf/.style={dot},int/.style={dot},]
\node[int] (u0) at (0,-0.5) {};
\node[int] (u1) at (-0.5,0) {};
\node[int] (u2) at (0,0.5) {};
\node[int] (u3) at (0.5,0) {};
\node[int] (u4) at (1.5,0) {};
\node[leaf, label=below:$A$] (A) at (0,-1) {};
\node[leaf, label=left:$B$] (B) at (-1, 0) {};
\node[leaf, label=above:$C$] (C) at (0,1) {};
\node[leaf, label=above right:$D$] (D) at (2,0.87) {};
\node[leaf, label=below right:$E$] (E) at (2,-0.87) {};
\draw[line width=0.8pt] (u3) -- (u0) -- (u1) -- (u2) -- (u3) -- (u4);
\draw[line width=0.8pt] (u0) -- (A);
\draw[line width=0.8pt] (u1) -- (B);
\draw[line width=0.8pt] (u2) -- (C);
\draw[line width=0.8pt] (u4) -- (D);
\draw[line width=0.8pt] (u4) -- (E);
\node[circle,fill=purple,inner sep=1.3pt] (root1) at (-.25,-.25) {};
\node[circle,fill=red,inner sep=1.3pt] (root2) at (.25,-.25) {};
\node[circle,fill=purple,inner sep=1.3pt] (root3) at (-.25,.25) {};
\node[circle,fill=red,inner sep=1.3pt] (root4) at (.25,.25) {};
\node[circle,fill=green,inner sep=1.3pt] (root5) at (0,-.75) {};
\node[circle,fill=green,inner sep=1.3pt] (root6) at (-.75,0) {};
\node[circle,fill=green,inner sep=1.3pt] (root7) at (0,.75) {};
\node[circle,fill=blue,inner sep=1.3pt] at (0.5,0) {};
\end{tikzpicture}
    \caption{The blue node is the hybrid node. When the root is on one of the green nodes, we are missing a single invariant of degree 13, when on one of the purple nodes, we are missing a degree 62 invariant, and on the red nodes we are missing a degree 29 invariant.}
    \label{fig:psuedowitness}
\end{figure}
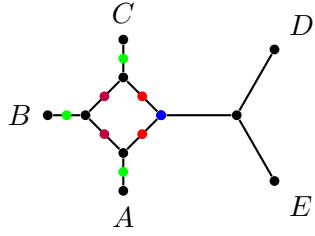

Case (2) of \Cref{thm:blob-tree-ii} is likely incomplete, though we are unable to prove this rigorously.
While we cannot certify that each of these cases are distinguishable from each other using symbolic calculations, 
we have ample numerical evidence that they are distinct.
We used psuedo-witness sets \cite{HauensteinPsuedoWitnessSets} to compute the degrees of each of these varieties numerically, with
the \texttt{Macaulay2} package \texttt{NumericalImplicitization} \cite{NumericalImplicitizationSource, NumericalImplicitizationArticle}.
The computation results in dividing the 7 networks into 3 classes, based on their degree.
There are 3 of degree 13, 2 of degree 62, and 2 of degree 29, determined by root location as
illustrated in \Cref{fig:psuedowitness}.

There is evidence that these 7 networks are all distinguishable from each other.
We ran the following experiment using the Julia package \texttt{HomotopyContinuation.jl} \cite{HomCont.jl}. 
For each network $\cn$ in \cref{fig:psuedowitness}, we input random real parameters into each of the parameterizations to get a numerical quintet CF vector for each network.
Then we set each of these CF vectors equal to the other parameterizations and attempted to solve this system of equations numerically.
We found no solutions, indicating our randomly chosen CF vectors do not lie
on the other CF varieties.
While this does not constitute a proof, the numerics strongly suggest the root location is identifiable when the hybrid in the 4-cycle is ancestral to $D$ and $E$.

As a last note on this case, our current approach is not feasible for distinguishing these cases. In each case, there is a single polynomial missing of degree 13, 29, or 62, depending on the root. Indeed, it can be seen in \cref{tab:table-blob-tree-ii-ideal-info} that the rank of the Jacobians for each of these networks is 7; moreover, in each of these cases the only invariants we have found are 7 linear ones. We may use these linear relations to reduce the number of variables from 15 to 8. Thus, in each case, we are looking for a prime ideal of dimension 7 in a polynomial ring in 8 variables of degree 13, 29, or 62 (depending on the network). The only such prime ideals are principal.
Finding these missing polynomials via interpolation methods will require truly massive amounts of computation. 

Finally,  while we have not computed the entire vanishing ideals for all 5-cycle networks for the third tree of blobs, we do obtain a complete identifiability result.

\begin{theorem}
\label{thm:blob-tree-iii}
    Let $\cn$ be a rooted level-1 binary network without 2-cycles whose tree of blobs is the tree on five taxa  with a 5-multifurcation, with the 4-blob having circular order shown in \Cref{fig:blob-trees} (iii).  Then the network root and hybrid locations are identifiable.
\end{theorem}

\begin{proof}
    Computing the ideals up to degree 4 distinguished most cases. However matroids were used to distinguish 2 cases, as shown in \Cref{tab:blob-tree-iii-ideal-info} and \Cref{fig:blob-trees-iii-rootings}
    of \Cref{app:TandF}.
\end{proof}

\section{Discussion and Future Work}
\label{sec: discussion}

In this work, we have taken steps 
to understanding the NMSC model via quintet concordance factors.
We introduced a recursive algorithm for computing
symbolic gene tree probabilities under the NMSC model.
Unlike past approaches,
rather than enumerating over all coalescent histories,
our approach is to reduce the network and gene tree by considering partial 
coalescent histories in a single edge.
While in this paper we focused on quintets,
the algorithm and implementation applies to networks and gene trees of arbitrary and differing sizes.
This is particularly useful in many cases
where obtaining the parameterizations by hand would be quite difficult. Additionally, it can be used for either symbolic or numerical calculations.

We applied this framework to all rooted binary level-1 phylogenetic networks on five taxa
without 2-cycles.
The resulting computations show that quintets carry
more information about the rooted network structure than quartets alone can. 
In almost all cases, the topological root location can be determined, up to location in a 3-cycle cluster (the cycle and two pendent edges).
For a network with a fully-resolved tree of blobs,
quintets can detect the presence of a central 3-cycle.
For networks with a 4-cycle, they identify the hybrid node 
and frequently constrain or completely determine the location of the root.
Interestingly, our results reveal non-identifiability phenomena as well. 
A 3-cycle located at either cherry in the tree of blobs may be invisible unless the root lies on the 3-cycle or 
on an edge pendant to the 3-cycle.
Thus, moving from quartets to quintets improves identifiability,
but does not eliminate all problems. We note, however, that there are, as yet, no stronger results on identifiability from unrooted topological gene trees of arbitrary size.

Several questions remain open. 
In many cases, our identifiability results rely on
low degree invariants or information about the underlying algebraic matroids.
It may be useful in the future for inference if the full ideals were better understood. In particular, an exploration of whether specific invariants can be tied to specific topological features of a network is desirable.
Semi-algebraic conditions (polynomial inequalities defining
the image of stochastic parameters) are also needed for complete understanding of identifiability.

The most immediate course of action is, of course, to address 2-cycles. This is needed, for instance, to be able to apply the quintet CF results here to networks on more than 5 taxa, as passing to induced 5-taxon networks can produce 2-cycles even when the full network has none.
We omitted any analysis of 2-cycles from this work, as once they are introduced the class of level-1 networks on 5-taxa becomes infinite, and new arguments are needed to supplement what can be obtained with computation.
While 2-cycles are not identifiable from quartet CFs, with quintets their  story is more complex, as we will develop in a forthcoming work.
Together with the results presented here, 
that study will provide a more full account of 
the rooted and small-cycle structure that can be recovered 
from quintet CFs under the NMSC on level-1 networks.

\section*{Declaration of generative AI and AI-assisted technologies in the manuscript preparation process.}
During the preparation of this work the authors used ChatGPT and Gemini as aids in producing illustrations. The authors reviewed and edited that content as needed and take full responsibility for the published article.

\section*{Acknowledgements}
We thank Dylan Alvarenga and Mary Hopkins who participated in the inception of this project. We also thank Benjamin Hollering for suggesting the use of algebraic matroids.

The work was begun at the Algebra of Phylogenetics Workshop '24 at the University of Hawai`i at M\=anoa funded by NSF grant DMS-1945584 to E. Gross, 
and continued at the Institute for Computational and Experimental Research in Mathematics in Providence, RI, under Grant DMS-1929284  during the Quantitative Phylogenetics  program, as well as during a Collaborate@ICERM visit.
J.A. Rhodes was partially supported by DMS-2051760.

\newcommand{\fiveTaxonTree}[5]{%
\begin{tikzpicture}[
    baseline=(current bounding box.center),
    scale=0.55,
    >=Stealth,
    every node/.style={font=\small},
    dot/.style={circle,fill=black,inner sep=1.3pt},
    leaf/.style={dot},
    int/.style={dot},
]
    \coordinate (u) at (-1,0);
    \coordinate (v) at (0,0);
    \coordinate (w) at (1,0);

    \node[leaf,label=left:$#1$]  (A) at (-1.5,-0.87) {};
    \node[leaf,label=left:$#2$]  (B) at (-1.5,0.87) {};
    \node[leaf,label=above:$#3$] (C) at (0,1) {};
    \node[leaf,label=right:$#4$] (D) at (1.5,0.87) {};
    \node[leaf,label=right:$#5$] (E) at (1.5,-0.87) {};

    \draw[line width=1.3pt]
        (A) -- (u) -- (B)
        (u) -- (v)
        (v) -- (C)
        (v) -- (w)
        (w) -- (D)
        (w) -- (E);
\end{tikzpicture}%
}

\renewcommand{\arraystretch}{1.3}

\appendix
\section{Tables and figures}\label{app:TandF}
\begin{table}[ht]
\centering
\begin{minipage}[t]{0.48\textwidth}
\vspace{0pt}
\centering
\begin{tabular}{c|c|c}
    & Splits & Graph  \\ \hline
    $T_1$ & $\substack{AB|CDE\\ ABC|DE}$ & \fiveTaxonTree{A}{B}{C}{D}{E} \\
    $T_2$ & $\substack{AB|CDE\\ ABD|CE}$ & \fiveTaxonTree{A}{B}{D}{C}{E}\\
    $T_3$ & $\substack{AB|CDE\\ ABE|CD}$ & \fiveTaxonTree{A}{B}{E}{C}{D} \\
    $T_4$ & $\substack{AC|BDE\\ ABC|DE}$ & \fiveTaxonTree{A}{C}{B}{D}{E}\\
    $T_5$ & $\substack{AC|BDE\\ ACD|BE}$ & \fiveTaxonTree{A}{C}{D}{B}{E}\\
    $T_6$ & $\substack{AC|BDE\\ ACE|BD}$ & \fiveTaxonTree{A}{C}{E}{B}{D}\\
    $T_7$ & $\substack{AD|BCE\\ ABD|CE}$ & \fiveTaxonTree{A}{D}{B}{C}{E}\\
    $T_8$ & $\substack{AD|BCE\\ ACD|BE}$ & \fiveTaxonTree{A}{D}{C}{B}{E}
\end{tabular}
\end{minipage}
\hfill
\begin{minipage}[t]{0.48\textwidth}
\vspace{0pt}
\centering
\begin{tabular}{c|c|c}
    & Splits & Graph \\ \hline
    $T_9$ & $\substack{AD|BCE\\ ADE|BC}$ & \fiveTaxonTree{A}{D}{E}{B}{C}\\
    $T_{10}$ & $\substack{AE|BCD\\ ABE|CD}$ & \fiveTaxonTree{A}{E}{B}{C}{D}\\
    $T_{11}$ & $\substack{AE|BCD\\ ACE|BD}$ & \fiveTaxonTree{A}{E}{C}{B}{D}\\
    $T_{12}$ & $\substack{AE|BCD\\ ADE|BC}$ & \fiveTaxonTree{A}{E}{D}{B}{C}\\
    $T_{13}$ & $\substack{BC|ADE\\ ABC|DE}$ & \fiveTaxonTree{B}{C}{A}{D}{E}\\
    $T_{14}$ & $\substack{BD|ACE\\ ABD|CE}$ & \fiveTaxonTree{B}{D}{A}{C}{E}\\
    $T_{15}$ & $\substack{BE|ACD\\ ABE|CD}$ & \fiveTaxonTree{B}{E}{A}{C}{D}
\end{tabular}
\end{minipage}
\caption{Ordering of the 15 gene trees with taxon set $\{A,B,C,D,E\}$.}
\label{tab:gene_tree_ordering}
\end{table}

In this appendix, we have various tables and figures referenced throughout the document.
\Cref{tab:gene_tree_ordering} enumerates the ordering we use on the fifteen 5-taxon gene trees.
\Cref{fig:blob-tree-i-rootings,fig:blob-tree-ii-rootings,fig:blob-trees-iii-rootings} enumerate all rooted level-1 networks under consideration in \Cref{sec: root identifiability}.
\Cref{tab:blob-trees-i-ideals,tab:table-blob-tree-ii-ideal-info,tab:blob-tree-iii-ideal-info} record summaries of our computational results needed to prove \Cref{thm:blob-tree-i,thm:blob-tree-ii,thm:blob-tree-iii}.
Computations verifying the entries in each table can be found on our \href{https://github.com/jcu237/SymbolicCoalescentModel}{GitHub} under the \texttt{ideals} directory.
Specifically, computations supporting \cref{tab:blob-trees-i-ideals,tab:table-blob-tree-ii-ideal-info,tab:blob-tree-iii-ideal-info} can be found in the following respective files.
\begin{itemize}
    \item \texttt{ideals/fullyResolvedBlobTreeIdeals.m2}
    \item \texttt{ideals/partiallyResolvedBlobTreeIdeals.m2}
    \item \texttt{ideals/unresolvedBlobTreeIdeals.m2}
\end{itemize}
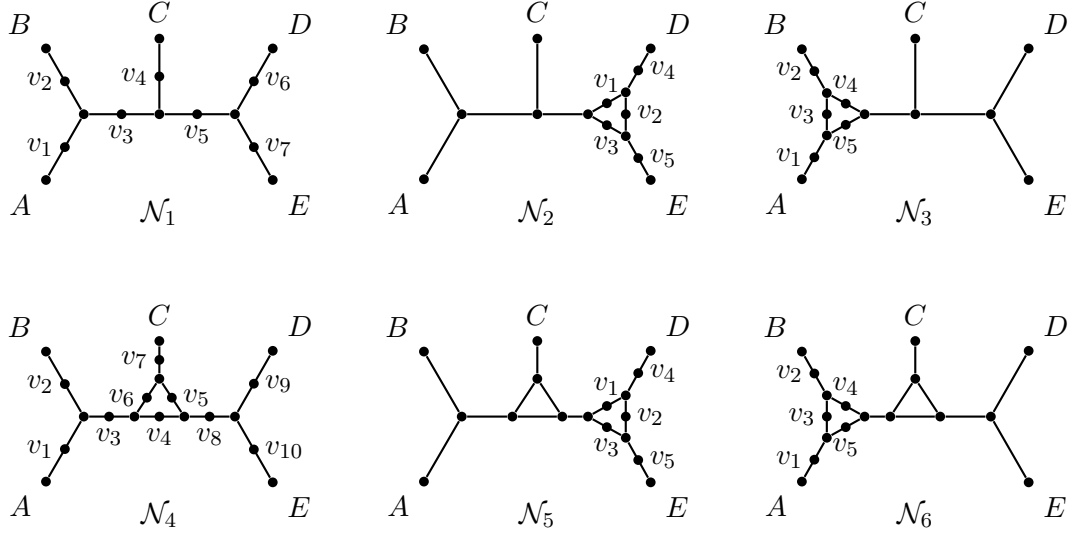
\begin{figure}[ht]
    \centering
    \begin{tikzpicture}[
        >=Stealth, 
        every node/.style={font=\small},
        dot/.style={circle,fill=black,inner sep=1.3pt},
        leaf/.style={dot},
        int/.style={dot},
    ]

    \begin{scope}[shift={(0,0)}]
        \node[int] (u) at (0,0) {};
        \node[int] (v) at (1,0) {};
        \node[int] (w) at (2,0) {};
        \node[leaf,label=below left:$A$] (A) at (-0.5,-0.87) {};
        \node[leaf,label=above left:$B$] (B) at (-0.5, 0.87) {};
        \node[leaf,label=above:$C$] (C) at (1,1) {};
        \node[leaf,label=above right:$D$] (D) at (2.5, 0.87) {};
        \node[leaf,label=below right:$E$] (E) at (2.5,-0.87) {};

        \draw[line width=0.8pt] (A) -- node[int,midway] {} node[midway,left] {$v_1$} (u);
        \draw[line width=0.8pt] (B) -- node[int,midway] {} node[midway,left] {$v_2$} (u);
        \draw[line width=0.8pt] (u) -- node[int,midway] {} node[midway,below] {$v_3$} (v);
        \draw[line width=0.8pt] (v) -- node[int,midway] {} node[midway,left] {$v_4$} (C);
        \draw[line width=0.8pt] (v) -- node[int,midway] {} node[midway,below] {$v_5$} (w);
        \draw[line width=0.8pt] (w) -- node[int,midway] {} node[midway,right] {$v_6$} (D);
        \draw[line width=0.8pt] (w) -- node[int,midway] {} node[midway,right] {$v_7$} (E);

        \node[below] at (1,-1) {$\cn_1$};
    \end{scope}

    \begin{scope}[shift={(5,0)}]
        \node[int] (u) at (0,0) {};
        \node[int] (v) at (1,0) {};
        \node[int] (w1) at (1.67,0) {};
        \node[int] (w2) at (2.17,0.29) {};
        \node[int] (w3) at (2.17,-0.29) {};
        \node[leaf, label=below left:$A$] (A) at (-.5, -0.86) {};
        \node[leaf, label=above left:$B$] (B) at (-.5, 0.86) {};
        \node[leaf, label=above:$C$] (C) at (1,1) {};
        \node[leaf, label=above right:$D$] (D) at (2.5,0.87) {};
        \node[leaf, label=below right:$E$] (E) at (2.5,-0.87) {};

        \draw[line width=0.8pt] (A) -- (u); 
        \draw[line width=0.8pt] (u) -- (B);
        \draw[line width=0.8pt] (u) -- (v); 
        \draw[line width=0.8pt] (v) -- (w1);
        \draw[line width=0.8pt] (C) -- (v);
        \draw[line width=0.8pt] (w1) -- node[int,midway] {} node[midway,above] {$v_1$} (w2);
        \draw[line width=0.8pt] (w2) -- node[int,midway] {} node[midway,right] {$v_2$} (w3);
        \draw[line width=0.8pt] (w3) -- node[int,midway] {} node[midway,below] {$v_3$} (w1);
        \draw[line width=0.8pt] (w2) -- node[int,midway] {} node[midway,right] {$v_4$} (D);
        \draw[line width=0.8pt] (w3) -- node[int,midway] {} node[midway,right] {$v_5$} (E);

        \node[below] at (1,-1) {$\cn_2$};
    \end{scope}

    \begin{scope}[shift={(10,0)}]
        \node[int] (u3) at (0.33,0) {};
        \node[int] (u1) at (-0.17,-0.28) {};
        \node[int] (u2) at (-0.17,0.28) {};
        \node[int] (v) at (1,0) {};
        \node[int] (w) at (2,0) {};
        \node[leaf, label=below left:$A$] (A) at (-.5, -0.86) {};
        \node[leaf, label=above left:$B$] (B) at (-.5, 0.86) {};
        \node[leaf, label=above:$C$] (C) at (1,1) {};
        \node[leaf, label=above right:$D$] (D) at (2.5,0.87) {};
        \node[leaf, label=below right:$E$] (E) at (2.5,-0.87) {};

        \draw[line width=0.8pt] (u1) -- node[int,midway] {} node[midway,left] {$v_3$} (u2);
        \draw[line width=0.8pt] (u2) -- node[int,midway] {} node[midway,above] {$v_4$} (u3);
        \draw[line width=0.8pt] (u3) -- node[int,midway] {} node[midway,below] {$v_5$} (u1);
        \draw[line width=0.8pt] (A) -- node[int,midway] {} node[midway,left] {$v_1$} (u1);
        \draw[line width=0.8pt] (B) -- node[int,midway] {} node[midway,left] {$v_2$} (u2);
        \draw[line width=0.8pt] (u3) -- (v);
        \draw[line width=0.8pt] (v) -- (w);
        \draw[line width=0.8pt] (C) -- (v);
        \draw[line width=0.8pt] (D) -- (w);
        \draw[line width=0.8pt] (w) -- (E);

        \node[below] at (1,-1) {$\cn_3$};
    \end{scope}

    \begin{scope}[shift={(0,-4)}]
        \node[int] (u) at (0,0) {};
        \node[int] (v1) at (0.67,0) {};
        \node[int] (v2) at (1,0.5) {};
        \node[int] (v3) at (1.33,0) {};
        \node[int] (w) at (2,0) {};
        \node[leaf, label=below left:$A$] (A) at (-.5, -0.86) {};
        \node[leaf, label=above left:$B$] (B) at (-.5, 0.86) {};
        \node[leaf, label=above:$C$] (C) at (1,1) {};
        \node[leaf, label=above right:$D$] (D) at (2.5,0.87) {};
        \node[leaf, label=below right:$E$] (E) at (2.5,-0.87) {};

        \draw[line width=0.8pt] (A) -- node[int,midway] {} node[midway,left] {$v_1$} (u);
        \draw[line width=0.8pt] (u) -- node[int,midway] {} node[midway,left] {$v_2$} (B);
        \draw[line width=0.8pt] (u) -- node[int,midway] {} node[midway,below] {$v_3$} (v1);
        \draw[line width=0.8pt] (v1) -- node[int,midway] {} node[midway,below] {$v_4$} (v3);
        \draw[line width=0.8pt] (v3) -- node[int,midway] {} node[midway,right] {$v_5$} (v2);
        \draw[line width=0.8pt] (v2) -- node[int,midway] {} node[midway,left] {$v_6$} (v1);
        \draw[line width=0.8pt] (C) -- node[int,midway] {} node[midway,left] {$v_7$} (v2);
        \draw[line width=0.8pt] (v3) -- node[int,midway] {} node[midway,below] {$v_8$} (w);
        \draw[line width=0.8pt] (D) -- node[int,midway] {} node[midway,right] {$v_9$} (w);
        \draw[line width=0.8pt] (w) -- node[int,midway] {} node[midway,right] {$v_{10}$} (E);

        \node[below] at (1,-1) {$\cn_4$};
    \end{scope}

    \begin{scope}[shift={(5,-4)}]
        \node[int] (u) at (0,0) {};
        \node[int] (v1) at (0.67,0) {};
        \node[int] (v2) at (1,0.5) {};  
        \node[int] (v3) at (1.33,0) {};
        \node[int] (w1) at (1.67,0) {};
        \node[int] (w2) at (2.17,0.28) {};
        \node[int] (w3) at (2.17,-0.28) {};
        \node[leaf, label=below left:$A$] (A) at (-.5, -0.86) {};
        \node[leaf, label=above left:$B$] (B) at (-.5, 0.86) {};
        \node[leaf, label=above:$C$] (C) at (1,1) {};
        \node[leaf, label=above right:$D$] (D) at (2.5,0.87) {};
        \node[leaf, label=below right:$E$] (E) at (2.5,-0.87) {};

        \draw[line width=0.8pt] (A) -- (u);
        \draw[line width=0.8pt] (u) -- (B);
        \draw[line width=0.8pt] (u) -- (v1);
        \draw[line width=0.8pt] (v1) -- (v2);
        \draw[line width=0.8pt] (v2) -- (v3);
        \draw[line width=0.8pt] (v3) -- (v1);
        \draw[line width=0.8pt] (C) -- (v2);
        \draw[line width=0.8pt] (v3) -- (w1);
        \draw[line width=0.8pt] (w1) -- node[int,midway] {} node[midway,above] {$v_1$} (w2);
        \draw[line width=0.8pt] (w2) -- node[int,midway] {} node[midway,right] {$v_2$} (w3);
        \draw[line width=0.8pt] (w3) -- node[int,midway] {} node[midway,below] {$v_3$} (w1);
        \draw[line width=0.8pt] (D) -- node[int,midway] {} node[midway,right] {$v_4$} (w2);
        \draw[line width=0.8pt] (w3) -- node[int,midway] {} node[midway,right] {$v_5$} (E);

        \node[below] at (1,-1) {$\cn_5$};
    \end{scope}

    \begin{scope}[shift={(10,-4)}]
        \node[int] (u3) at (0.33,0) {};
        \node[int] (u1) at (-0.17,-0.28) {};
        \node[int] (u2) at (-0.17,0.28) {};
        \node[int] (v1) at (0.67,0) {};
        \node[int] (v2) at (1.33,0) {};
        \node[int] (v3) at (1,0.5) {};
        \node[int] (w) at (2,0) {};
        \node[leaf, label=below left:$A$] (A) at (-.5, -0.86) {};
        \node[leaf, label=above left:$B$] (B) at (-.5, 0.86) {};
        \node[leaf, label=above:$C$] (C) at (1,1) {};
        \node[leaf, label=above right:$D$] (D) at (2.5,0.87) {};
        \node[leaf, label=below right:$E$] (E) at (2.5,-0.87) {};
        
        \draw[line width=0.8pt] (u1) -- node[int,midway] {} node[midway,left] {$v_3$} (u2);
        \draw[line width=0.8pt] (u2)-- node[int,midway] {} node[midway,above] {$v_4$} (u3);
        \draw[line width=0.8pt] (u3) -- node[int,midway] {} node[midway,below] {$v_5$} (u1);
        \draw[line width=0.8pt] (A) -- node[int,midway] {} node[midway,left] {$v_1$} (u1);
        \draw[line width=0.8pt] (B) -- node[int,midway] {} node[midway,left] {$v_2$} (u2);
        \draw[line width=0.8pt] (u3) -- (v1);
        \draw[line width=0.8pt] (v1) -- (v2);
        \draw[line width=0.8pt] (v2) -- (v3);
        \draw[line width=0.8pt] (v3) -- (v1);
        \draw[line width=0.8pt] (C) -- (v3);
        \draw[line width=0.8pt] (v2) -- (w);
        \draw[line width=0.8pt] (D) -- (w);
        \draw[line width=0.8pt] (w) -- (E);

        \node[below] at (1,-1) {$\cn_6$};
    \end{scope}
    \end{tikzpicture}
    \caption{These six networks above are representatives of the networks with Tree of Blobs from \Cref{fig:blob-trees} (i) that are identifiable using quintet CFs.
    They are undirected; however, their possible roots are labelled on each graph. The possible hybrid node locations are not labelled since in any 3-cycle, once a root is chosen,
    the hybrid has two or three possible locations, none of which are identifiable using quintet CFs. Note that all cases where a 3-cycle is located on a cherry are omitted.}
    \label{fig:blob-tree-i-rootings}
\end{figure}

\begin{table}[ht]
    \centering
    \begin{tabular}{c|c|c|c}
        Networks & $\dim(V_\cn)$ & $\deg(V_\cn)$  & $(d_1,d_2,d_3)$ \\ \hline
        $\cn_1(v_1)$ & 3 & 4 & $(10,1,3)$ \\
        $\cn_1(v_2)$ & 3 & 4 & $(10,1,3)$ \\
        $\cn_1(v_3)$ & 3 & 2 & $(11,1,0)$ \\
        $\cn_1(v_4)$ & 3 & 2 & $(11,1,0)$ \\
        $\cn_1(v_5)$ & 3 & 2 & $(11,1,0)$ \\
        $\cn_1(v_6)$ & 3 & 4 & $(10,1,3)$ \\
        $\cn_1(v_7)$ & 3 & 4 & $(10,1,3)$ \\
        $\{\cn_2(v_i)~|~1\leq i\leq 5\}$ & 4 & 2 & $(10,1,0)$ \\
        $\{\cn_3(v_i)~|~1\leq i\leq 5\}$ & 4  & 2 & $(10,1,0)$ \\
        $\cn_4(v_3)$ & 4 & 1 & $(11,0,0)$ \\
        $\{\cn_4(v_i)~|~4\leq i \leq 7\}$ & 5 & 1 & $(10,0,0)$ \\
        $\cn_4(v_8)$ & 4 & 1 & $(11,0,0)$\\
        $\substack{\{\cn_4(v_i) ~|~ i=9,10\}\\ \{\cn_5(v_i) ~|~ 1 \leq i\leq 5\}}$ & 5 & 1 & $(10,0,0)$ \\
        $\substack{\{\cn_4(v_i)~|~i=1,2\}\\\{\cn_6(v_i) ~|~ 1 \leq i\leq 5\}}$ & 5 & 1 & $(10,0,0)$ \\
    \end{tabular}
    \caption{Each row gives a collection of networks with a common and unique quintet CF ideal whose generators can be taken to have degree at most 3. The dimension and degree of the CF varieties are recorded. The last column records the minimal number of invariants of degree $i$ needed to generate the ideal where $d_i$ is the number of degree $i$ invariants needed.}
    \label{tab:blob-trees-i-ideals}
\end{table}

\begin{figure}[ht]
    \centering
    \begin{tikzpicture}[
        >=Stealth, 
        every node/.style={font=\small},
        dot/.style={circle,fill=black,inner sep=1.3pt},
        leaf/.style={dot},
        int/.style={dot},
    ]

    \begin{scope}[shift={(0,0)}]
        \node[circle,fill=blue,inner sep=1.3pt] (u0) at (0,-0.5) {};
        \node[int] (u1) at (-0.5,0) {};
        \node[int] (u2) at (0,0.5) {};
        \node[int] (u3) at (0.5,0) {};
        \node[int] (u4) at (2,0) {};
        \node[leaf, label=below:$A$] (A) at (0,-1) {};
        \node[leaf, label=left:$B$] (B) at (-1.5, 0) {};
        \node[leaf, label=above:$C$] (C) at (0,1.5) {};
        \node[leaf, label=above right:$D$] (D) at (2.5,0.87) {};
        \node[leaf, label=below right:$E$] (E) at (2.5,-0.87) {};
        
        \draw[line width=0.8pt] (u3) -- node[int,midway] {} node[midway,below right] {$v_6$} (u0);
        \draw[line width=0.8pt] (u0) -- node[int,midway] {} node[midway,below left] {$v_7$} (u1);
        \draw[line width=0.8pt] (u1) -- node[int,midway] {} node[midway,above left] {$v_8$} (u2);
        \draw[line width=0.8pt] (u2) -- node[int,midway] {} node[midway,above right] {$v_9$} (u3);
        \draw[line width=0.8pt] (u3) -- node[int,midway] {} node[midway,below] {$v_{10}$} (u4);
        \draw[line width=0.8pt] (u0) -- (A);
        \draw[line width=0.8pt] (u1) -- node[int,midway] {} node[midway,above] {$v_2$} (B);
        \draw[line width=0.8pt] (u2) -- node[int,midway] {} node[midway,left] {$v_3$} (C);
        \draw[line width=0.8pt] (u4) -- node[int,midway] {} node[midway,right] {$v_4$} (D);
        \draw[line width=0.8pt] (u4) -- node[int,midway] {} node[midway,right] {$v_5$} (E);

        \node[below] at (1,-1.5) {$\cn_7$};
    \end{scope}

    \begin{scope}[shift={(6,0)}]
        \node[int] (u0) at (0,-0.5) {};
        \node[circle,fill=blue,inner sep=1.3pt] (u1) at (-0.5,0) {};
        \node[int] (u2) at (0,0.5) {};
        \node[int] (u3) at (0.5,0) {};
        \node[int] (u4) at (2,0) {};
        \node[leaf, label=below:$A$] (A) at (0,-1.5) {};
        \node[leaf, label=left:$B$] (B) at (-1, 0) {};
        \node[leaf, label=above:$C$] (C) at (0,1.5) {};
        \node[leaf, label=above right:$D$] (D) at (2.5,0.87) {};
        \node[leaf, label=below right:$E$] (E) at (2.5,-0.87) {};
        
        \draw[line width=0.8pt] (u3) -- node[int,midway] {} node[midway,below right] {$v_6$} (u0);
        \draw[line width=0.8pt] (u0) -- node[int,midway] {} node[midway,below left] {$v_7$} (u1);
        \draw[line width=0.8pt] (u1) -- node[int,midway] {} node[midway,above left] {$v_8$} (u2);
        \draw[line width=0.8pt] (u2) -- node[int,midway] {} node[midway,above right] {$v_9$} (u3);
        \draw[line width=0.8pt] (u3) -- node[int,midway] {} node[midway,below] {$v_{10}$} (u4);
        \draw[line width=0.8pt] (u0) -- node[int,midway] {} node[midway,left] {$v_1$} (A);
        \draw[line width=0.8pt] (u1) -- (B);
        \draw[line width=0.8pt] (u2) -- node[int,midway] {} node[midway,left] {$v_3$} (C);
        \draw[line width=0.8pt] (u4) -- node[int,midway] {} node[midway,right] {$v_4$} (D);
        \draw[line width=0.8pt] (u4) -- node[int,midway] {} node[midway,right] {$v_5$} (E);
        \node[below] at (1,-1.5) {$\cn_8$};
    \end{scope}

    \begin{scope}[shift={(12,0)}]
        \node[int] (u0) at (0,-0.5) {};
        \node[int] (u1) at (-0.5,0) {};
        \node[circle,fill=blue,inner sep=1.3pt] (u2) at (0,0.5) {};
        \node[int] (u3) at (0.5,0) {};
        \node[int] (u4) at (2,0) {};
        \node[leaf, label=below:$A$] (A) at (0,-1.5) {};
        \node[leaf, label=left:$B$] (B) at (-1.5, 0) {};
        \node[leaf, label=above:$C$] (C) at (0,1) {};
        \node[leaf, label=above right:$D$] (D) at (2.5,0.87) {};
        \node[leaf, label=below right:$E$] (E) at (2.5,-0.87) {};
        
        \draw[line width=0.8pt] (u3) -- node[int,midway] {} node[midway,below right] {$v_6$} (u0);
        \draw[line width=0.8pt] (u0) -- node[int,midway] {} node[midway,below left] {$v_7$} (u1);
        \draw[line width=0.8pt] (u1) -- node[int,midway] {} node[midway,above left] {$v_8$} (u2);
        \draw[line width=0.8pt] (u2) -- node[int,midway] {} node[midway,above right] {$v_9$} (u3);
        \draw[line width=0.8pt] (u3) -- node[int,midway] {} node[midway,below] {$v_{10}$} (u4);
        \draw[line width=0.8pt] (u0) -- node[int,midway] {} node[midway,left] {$v_1$} (A);
        \draw[line width=0.8pt] (u1) -- node[int,midway] {} node[midway,above] {$v_2$} (B);
        \draw[line width=0.8pt] (u2) -- (C);
        \draw[line width=0.8pt] (u4) -- node[int,midway] {} node[midway,right] {$v_4$} (D);
        \draw[line width=0.8pt] (u4) -- node[int,midway] {} node[midway,right] {$v_5$} (E);
        \node[below] at (1,-1.5) {$\cn_9$};
    \end{scope}

    \begin{scope}[shift={(0,-5)}]
        \node[int] (u0) at (0,-0.5) {};
        \node[int] (u1) at (-0.5,0) {};
        \node[int] (u2) at (0,0.5) {};
        \node[circle,fill=blue,inner sep=1.3pt] (u3) at (0.5,0) {};
        \node[int] (u4) at (1.5,0) {};
        \node[leaf, label=below:$A$] (A) at (0,-1.5) {};
        \node[leaf, label=left:$B$] (B) at (-1.5, 0) {};
        \node[leaf, label=above:$C$] (C) at (0,1.5) {};
        \node[leaf, label=above right:$D$] (D) at (2,0.87) {};
        \node[leaf, label=below right:$E$] (E) at (2,-0.87) {};
        
        \draw[line width=0.8pt] (u3) -- node[int,midway] {} node[midway,below right] {$v_6$} (u0);
        \draw[line width=0.8pt] (u0) -- node[int,midway] {} node[midway,below left] {$v_7$} (u1);
        \draw[line width=0.8pt] (u1) -- node[int,midway] {} node[midway,above left] {$v_8$} (u2);
        \draw[line width=0.8pt] (u2) -- node[int,midway] {} node[midway,above right] {$v_9$} (u3);
        \draw[line width=0.8pt] (u3) -- (u4);
        \draw[line width=0.8pt] (u0) -- node[int,midway] {} node[midway,left] {$v_1$} (A);
        \draw[line width=0.8pt] (u1) -- node[int,midway] {} node[midway,above] {$v_2$} (B);
        \draw[line width=0.8pt] (u2) -- node[int,midway] {} node[midway,left] {$v_3$} (C);
        \draw[line width=0.8pt] (u4) -- (D);
        \draw[line width=0.8pt] (u4) -- (E);
        \node[below] at (1,-1.5) {$\cn_{10}$};
    \end{scope}

    \begin{scope}[shift={(6,-5)}]
        \node[circle,fill=blue,inner sep=1.3pt] (u0) at (0,-0.5) {};
        \node[int] (u1) at (-0.5,0) {};
        \node[int] (u2) at (0,0.5) {};
        \node[int] (u3) at (0.5,0) {};
        \node[int] (u4) at (1.17,0) {};
        \node[int] (u5) at (1.66,0.29) {};
        \node[int] (u6) at (1.66,-0.29) {};
        \node[leaf, label=below:$A$] (A) at (0,-1) {};
        \node[leaf, label=left:$B$] (B) at (-1, 0) {};
        \node[leaf, label=above:$C$] (C) at (0,1) {};
        \node[leaf, label=above right:$D$] (D) at (2,0.87) {};
        \node[leaf, label=below right:$E$] (E) at (2,-0.87) {};
        \draw[line width=0.8pt] (u3) -- (u0);
        \draw[line width=0.8pt] (u0) -- (u1);
        \draw[line width=0.8pt] (u1) -- (u2);
        \draw[line width=0.8pt] (u2) -- (u3);
        \draw[line width=0.8pt] (u3) -- (u4);
        \draw[line width=0.8pt] (u4) -- node[int,midway] {} node[midway,above left] {$v_1$} (u5);
        \draw[line width=0.8pt] (u5) -- node[int,midway] {} node[midway,right] {$v_2$} (u6);
        \draw[line width=0.8pt] (u6) -- node[int,midway] {} node[midway,below left] {$v_3$} (u4);
        \draw[line width=0.8pt] (u0) -- (A);
        \draw[line width=0.8pt] (u1) -- (B);
        \draw[line width=0.8pt] (u2) -- (C);
        \draw[line width=0.8pt] (u5) -- node[int,midway] {} node[midway,right] {$v_4$} (D);
        \draw[line width=0.8pt] (u6) -- node[int,midway] {} node[midway,right] {$v_5$} (E);
        \node[below] at (1,-1.5) {$\cn_{11}$};
    \end{scope}

        \begin{scope}[shift={(12,-5)}]
        \node[int] (u0) at (0,-0.5) {};
        \node[circle,fill=blue,inner sep=1.3pt] (u1) at (-0.5,0) {};
        \node[int] (u2) at (0,0.5) {};
        \node[int] (u3) at (0.5,0) {};
        \node[int] (u4) at (1.17,0) {};
        \node[int] (u5) at (1.66,0.29) {};
        \node[int] (u6) at (1.66,-0.29) {};
        \node[leaf, label=below:$A$] (A) at (0,-1) {};
        \node[leaf, label=left:$B$] (B) at (-1, 0) {};
        \node[leaf, label=above:$C$] (C) at (0,1) {};
        \node[leaf, label=above right:$D$] (D) at (2,0.87) {};
        \node[leaf, label=below right:$E$] (E) at (2,-0.87) {};
        \draw[line width=0.8pt] (u3) -- (u0);
        \draw[line width=0.8pt] (u0) -- (u1);
        \draw[line width=0.8pt] (u1) -- (u2);
        \draw[line width=0.8pt] (u2) -- (u3);
        \draw[line width=0.8pt] (u3) -- (u4);
        \draw[line width=0.8pt] (u4) -- node[int,midway] {} node[midway,above left] {$v_1$} (u5);
        \draw[line width=0.8pt] (u5) -- node[int,midway] {} node[midway,right] {$v_2$} (u6);
        \draw[line width=0.8pt] (u6) -- node[int,midway] {} node[midway,below left] {$v_3$} (u4);
        \draw[line width=0.8pt] (u0) -- (A);
        \draw[line width=0.8pt] (u1) -- (B);
        \draw[line width=0.8pt] (u2) -- (C);
        \draw[line width=0.8pt] (u5) -- node[int,midway] {} node[midway,right] {$v_4$} (D);
        \draw[line width=0.8pt] (u6) -- node[int,midway] {} node[midway,right] {$v_5$} (E);
        \node[below] at (1,-1.5) {$\cn_{12}$};
    \end{scope}

    \begin{scope}[shift={(6,-10)}]
        \node[int] (u0) at (0,-0.5) {};
        \node[int] (u1) at (-0.5,0) {};
        \node[circle,fill=blue,inner sep=1.3pt] (u2) at (0,0.5) {};
        \node[int] (u3) at (0.5,0) {};
        \node[int] (u4) at (1.17,0) {};
        \node[int] (u5) at (1.66,0.29) {};
        \node[int] (u6) at (1.66,-0.29) {};
        \node[leaf, label=below:$A$] (A) at (0,-1) {};
        \node[leaf, label=left:$B$] (B) at (-1, 0) {};
        \node[leaf, label=above:$C$] (C) at (0,1) {};
        \node[leaf, label=above right:$D$] (D) at (2,0.87) {};
        \node[leaf, label=below right:$E$] (E) at (2,-0.87) {};
        \draw[line width=0.8pt] (u3) -- (u0);
        \draw[line width=0.8pt] (u0) -- (u1);
        \draw[line width=0.8pt] (u1) -- (u2);
        \draw[line width=0.8pt] (u2) -- (u3);
        \draw[line width=0.8pt] (u3) -- (u4);
        \draw[line width=0.8pt] (u4) -- node[int,midway] {} node[midway,above left] {$v_1$} (u5);
        \draw[line width=0.8pt] (u5) -- node[int,midway] {} node[midway,right] {$v_2$} (u6);
        \draw[line width=0.8pt] (u6) -- node[int,midway] {} node[midway,below left] {$v_3$} (u4);
        \draw[line width=0.8pt] (u0) -- (A);
        \draw[line width=0.8pt] (u1) -- (B);
        \draw[line width=0.8pt] (u2) -- (C);
        \draw[line width=0.8pt] (u5) -- node[int,midway] {} node[midway,right] {$v_4$} (D);
        \draw[line width=0.8pt] (u6) -- node[int,midway] {} node[midway,right] {$v_5$} (E);
        \node[below] at (1,-1.5) {$\cn_{13}$};
    \end{scope}
    
    \end{tikzpicture}
    \caption{These seven networks above are representatives of the networks with Tree of Blobs from \Cref{fig:blob-trees} (ii) that are identifiable using quintet CFs.
    The possible root locations are marked in each, and the hybrid location in the 4-cycle is marked in blue.
    Note that not all rootings are listed on networks 11-14 since these CF varieties are already covered in previous cases.}
    \label{fig:blob-tree-ii-rootings}
\end{figure}
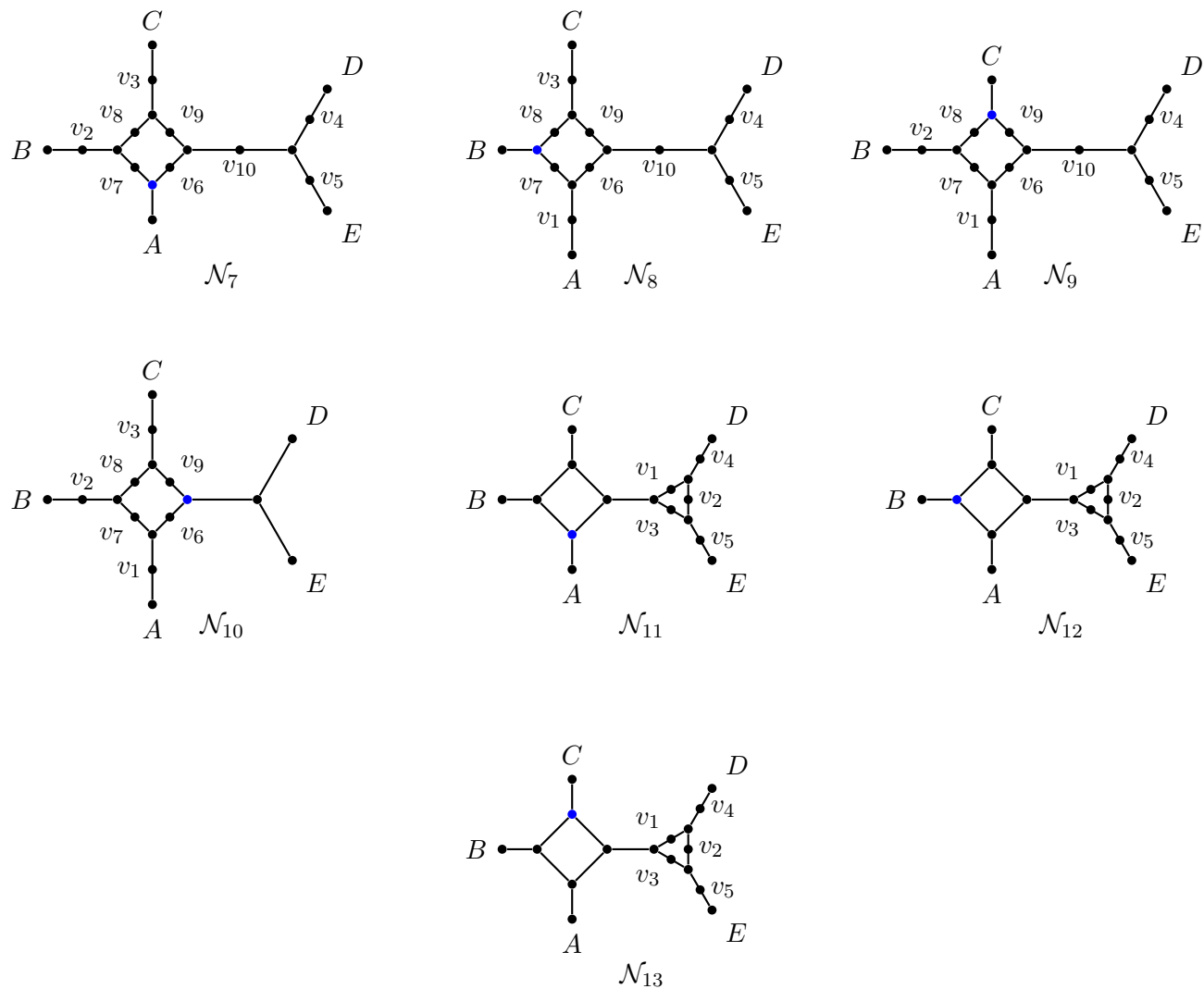

\begin{table}[ht]
    \centering
    \begin{tabular}[]{c|c|c|c|c}
        Networks & $\dim(V_\cn)$  & $\deg(V_\cn)$  & $(d_1,d_2,\dotsc)$ & Circuits \\ \hline
        $\cn_7(v_2)$ & 5 & $20^{(n)}$  &  $(7,1,0,4,\dots)$ & $-$\\
        $\cn_7(v_3)$ & 5 & 6  &  $(8,1,0,3)$ & $-$\\
        $\cn_7(v_4)$ & 5 & 6  &  $(8,1,0,3)$ & $-$\\
        $\cn_7(v_5)$ & 5 & 6  &  $(8,1,0,3)$ & $-$\\
        $\cn_7(v_6)$ & 5 & 8  &  $(8,1,0,0,3)$ & $-$\\
        $\cn_7(v_7), \cn_7(v_8)$ & 5 & $48^{(n)}, 30^{(n)}$  & $(7,1,0,3,\dotsc)$ & $\substack{\{4,6,10,11,12,13\}\\ \in \mathcal{C}(M_8) \setminus \mathcal{C}(M_7)}$\\
        $\cn_7(v_9)$ & 5 & 11  & $(8,1,0,0,0,0,4)$ & $-$\\
        $\cn_7(v_{10})$ & 5 & 2  & $(9,1,0,0)$ & $-$\\
        
        $\cn_8(v_1),\cn_8(v_6), \cn_8(v_7)$ & 5  & 3  & $(8,3,0,0)$ & $-$\\
        $\cn_8(v_3),\cn_8(v_8),\cn_8(v_9) $ & 5 & 3 & $(8,3,0,0)$ & $-$\\
        $\cn_8(v_{10})$ & 4 & 3  & $(9,3,0,0)$ & $-$ \\
        $\cn_8(v_4) $ & 4 & 5 & $(8,3,6,0)$ & $-$\\
        $\cn_8(v_5) $ & 4 & 5 & $(8,3,6,0)$ & $-$ \\

        $\cn_9(v_1)$ & 5 & 6 & $(8,1,0,3)$ & $-$ \\
        $\cn_9(v_2)$ & 5 & $20^{(n)}$ & $(7,1,0,4,\dotsc)$ & $-$ \\
        $\cn_9(v_4)$ & 5 & 6 & $(8,1,0,3)$ & $-$ \\
        $\cn_9(v_5)$ & 5 & 6 & $(8,1,0,3)$ & $-$ \\
        $\cn_9(v_6)$ & 5 & 11 & $(8,1,0,0,0,0,4)$ & $-$ \\
        $\cn_9(v_7),\cn_9(v_8)$ & 5 & $30^{(n)}, 48^{(n)}$ & $(7,1,0,3,\dotsc)$ & $\substack{\{1,3,4,6,10,11\}\\ \in \mathcal{C}(M_7) \setminus \mathcal{C}(M_8)}$ \\
        $\cn_9(v_9)$ & 5 & 8 & $(8,2,0,0,3)$ & $-$ \\
        $\cn_9(v_{10})$ & 5 & 2 & $(9,1,0,0)$ & $-$ \\

        $\{\cn_{10}(v_i) ~|~\substack{1 \leq i \leq 3 \\ 6 \leq i \leq 9}\}$ & 7 & $\substack{13^{(n)},13^{(n)},13^{(n)},\\29^{(n)},62^{(n)},62^{(n)},29^{(n)}}$ & $(7,0,0,0,\dotsc)$ & $-$ \\

        $\{\cn_{11}(v_i) ~|~ 1\leq i \leq 5\}$ & 6 & 2 & $(8,1,0,0)$ & $-$ \\
        $\{\cn_{12}(v_i) ~|~ 1\leq i \leq 5\}$ & 5 & 3 & $(8,3,0,0)$ & $-$ \\
        $\{\cn_{13}(v_i) ~|~ 1\leq i \leq 5\}$ & 6 & 2 & $(8,1,0,0)$ & $-$ \\
        
    \end{tabular}
    \caption{Each row gives a collection of networks with a common and unique quintet CF ideal. The dimension and degree of the CF varieties are recorded. The the superscript $(n)$ indicates that the degree was computed numerically. The fourth column records the minimal number of invariants of degree $i$ needed to generate the ideal where $d_i$ is the number of degree $i$ invariants. An entry containing ``$\dotsc$'' indicates that the full vanishing ideal is not yet known. The final column records distinguishing circuits in the algebraic matroids if these were used to distinguish the networks. }
    \label{tab:table-blob-tree-ii-ideal-info}
\end{table}

\begin{figure}[ht]
    \centering
    \begin{tikzpicture}[
        >=Stealth, 
        every node/.style={font=\small},
        dot/.style={circle,fill=black,inner sep=1.3pt},
        leaf/.style={dot},
        int/.style={dot},
    ]

    \node[circle,fill=blue,inner sep=1.3pt] (u0) at (1,0) {};
    \node[int] (u1) at (0.31,0.95) {};
    \node[int] (u2) at (-0.81,0.59) {};
    \node[int] (u3) at (-0.81,-0.59) {};
    \node[int] (u4) at (0.31,-0.95) {};

    \node[leaf, label=right:$A$] (A) at (2,0) {};
    \node[leaf, label=above:$B$] (B) at (0.62, 1.9) {};
    \node[leaf, label=above left:$C$] (C) at (-1.62,1.18) {};
    \node[leaf, label=below left:$D$] (D) at (-1.62,-1.18) {};
    \node[leaf, label=below:$E$] (E) at (0.62,-1.9) {};

    \draw[line width=0.8pt] (u0) -- node[int,midway] {} node[midway,right] {$v_1$} (u1);
    \draw[line width=0.8pt] (u1) -- node[int,midway] {} node[midway,above] {$v_2$} (u2);
    \draw[line width=0.8pt] (u2) -- node[int,midway] {} node[midway,left] {$v_3$} (u3);
    \draw[line width=0.8pt] (u3) -- node[int,midway] {} node[midway,above] {$v_4$} (u4);
    \draw[line width=0.8pt] (u0) -- node[int,midway] {} node[midway,right] {$v_5$} (u4);

    \draw[line width=0.8pt] (u0) -- (A);
    \draw[line width=0.8pt] (u1) -- node[int,midway] {} node[midway,above left] {$v_6$} (B);
    \draw[line width=0.8pt] (u2) -- node[int,midway] {} node[midway,above] {$v_7$} (C);
    \draw[line width=0.8pt] (u3) -- node[int,midway] {} node[midway,above left] {$v_8$} (D);
    \draw[line width=0.8pt] (u4) -- node[int,midway] {} node[midway,left] {$v_9$} (E);

    \node[below] at (0,-2.5) {$\cn_{14}$};
    \end{tikzpicture}
    \caption{Theese are the 9 networks with tree of blobs (iii) under consideration. The hybrid is always ancestral to $A$, and the root can be placed at any one of the $v_i$'s.}
    \label{fig:blob-trees-iii-rootings}
\end{figure}
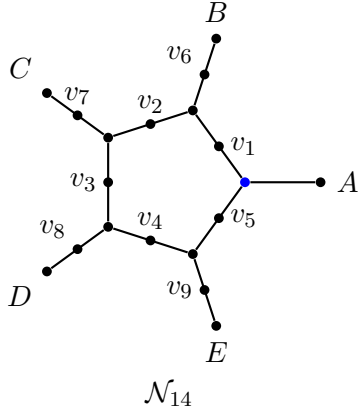

\begin{table}[ht]
    \centering
    \begin{tabular}{c|c|c|c|c}
        Networks & $\dim(V_\cn)$ & $\deg(V_\cn)$ & $(d_1,d_2,\dotsc)$ & Circuits \\ \hline
        $\cn_{14}(v_1),\cn_{14}(v_2)$ & 5 & $35^{(n)}$, $20^{(n)}$ & $(6,3,6,0,\dotsc)$ & $\substack{\{6,8,9,10,11,13\}\\ \in \mathcal{C}(M_2) \setminus \mathcal{C}(M_1)}$ \\
        $\cn_{14}(v_3)$ & 5 & 9 & $(7,3,0,0,10,1,1)$ & $-$ \\
        $\cn_{14}(v_4),\cn_{14}(v_5)$ & 5 & $20^{(n)}$, $35^{(n)}$ & $(6,3,6,0,\dotsc)$ & $\substack{\{1,3,5,6,8,9\}\\ \in\mathcal{C}(M_4) \setminus \mathcal{C}(M_5)}$\\
        $\cn_{14}(v_6)$ & 5 & $15^{(n)}$ & $(6,3,7,0,\dotsc)$ & $-$ \\
        $\cn_{14}(v_7)$ & 5 & 5 & $(7,3,6,0)$ & $-$\\
        $\cn_{14}(v_8)$ & 5 & 5 & $(7,3,6,0)$ & $-$ \\
        $\cn_{14}(v_9)$ & 5 & $15^{(n)}$ & $(6,3,7,0,\dotsc)$ & $-$
    \end{tabular}
    \caption{Each row gives a collection of networks with a common and unique quintet CF ideal. The dimension and degree of the CF varieties are recorded. The the superscript $(n)$ indicates that the degree was computed numerically. The fourth column records the minimal number of invariants of degree $i$ needed to generate the ideal where $d_i$ is the number of degree $i$ invariants. An entry containing ``$\dotsc$'' indicates that the full vanishing ideal is not yet known. The final column records distinguishing circuits in the algebraic matroids if these were used to distinguish the networks. 
     Finally, for the networks in rows 4 and 7, it is possible the full ideals were computed. The ideals generated have the correct dimension and the degrees of the partially generated ideals match the numerically computed degrees; however, the primality test failed to finish, so we were not able to verify the full vanishing ideal was computed. 
    }
    \label{tab:blob-tree-iii-ideal-info}
\end{table}

\clearpage

\bibliography{Hybridization}
\bibliographystyle{alpha}
\end{document}